\documentclass[journal]{IEEEtran}
\usepackage{cite}
\usepackage{url}

\usepackage{textcomp}
\usepackage{xcolor}
\usepackage{hyperref}

\usepackage{booktabs}	
\usepackage{diagbox}    
\usepackage{multirow}   

\usepackage{verbatim}	

\usepackage{amsmath,amssymb,amsfonts,graphicx}
\usepackage{epstopdf}

\newtheorem{definition}{Definition}[section]

\newenvironment{proof}{{\noindent\it\bf Proof}\quad}{\par}
\newtheorem{theorem}{Theorem}   

\usepackage{graphicx,epstopdf,algpseudocode,caption,url}   
\usepackage[ruled,linesnumbered]{algorithm2e}

\usepackage{amssymb}
\usepackage{hyperref}

\usepackage[switch]{lineno}   

\ifCLASSINFOpdf

\else

\fi

\begin{document}

\title{FUIM: Fuzzy Utility Itemset Mining}

\author{Shicheng Wan, Wensheng Gan,~\IEEEmembership{Member,~IEEE,} Xu Guo, Jiahui Chen,~\IEEEmembership{Member,~IEEE,} and Unil Yun

\thanks{This work was partially supported by the National Natural Science Foundation of China (Grant Nos. 61902079 and 62002136), and Guangzhou Basic and Applied Basic Research Foundation (Grant Nos. 202102020277  and 202102020928). (Corresponding author: Wensheng Gan and Jiahui Chen)}
	
	\thanks{Shicheng Wan, Xu Guo, and Jiahui Chen  are with the School of Computers, Guangdong University of Technology, Guangzhou 510006, China. (E-mail: scwan1998@gmail.com, csxuguo@gmail.com, and csjhchen@gmail.com)} 
		
	\thanks{Wensheng Gan is with the College of Cyber Security, Jinan University, Guangzhou 510632, China; and with Pazhou Laboratory, Guangzhou 510335, China. (E-mail: wsgan001@gmail.com)}	
	
	\thanks{Unil Yun is with Department of Computer Engineering, Sejong University, Seoul, South Korea (E-mail: yunei@sejong.ac.kr)}

}

\maketitle

\begin{abstract}

	Because of usefulness and comprehensibility, fuzzy data mining has been extensively studied and is an emerging topic in recent years. Compared with  utility-driven itemset mining technologies, fuzzy utility mining not only takes utilities (e.g., profits) into account, but also considers quantities of items in each transaction for discovering high fuzzy utility itemsets (HFUIs). Thus, fuzziness can be regard as a key criterion to select high-utility itemsets, while the exiting algorithms are not efficient enough. In this paper, an efficient one-phase algorithm named Fuzzy-driven Utility Itemset Miner (FUIM) is proposed to find out a complete set of HFUIs effectively. In addition, a novel compact data structure named fuzzy-list keeps the key information from quantitative transaction databases. Using fuzzy-list, FUIM can discover HFUIs from transaction databases efficiently and effectively. Both completeness and correctness of the FUIM algorithm are proved by five theorems. At last, substantial experiments test three terms (runtime cost, memory consumption, and scalability) to confirm that FUIM considerably outperforms the state-of-the-art algorithms.
\end{abstract}

\begin{IEEEkeywords}
	quantitative database, fuzzy-list, fuzzy theory, high fuzzy utility itemset.
\end{IEEEkeywords}

\IEEEpeerreviewmaketitle

\section{Introduction}
\label{sec:introduction}

Association rule mining (ARM) \cite{karthikeyan2014survey}  is a traditional method data mining technique which has been widely applied in many real applications. ARM aims to discover the inner link of frequency and confidence of itemsets from a set of data. In this framework, ARM algorithms try to extract frequent (aka high co-occurrence)  patterns, such as a married man who often buys diapers might also carries a dozen of beer. Another interesting task, a subfield of ARM, frequent itemset mining (FIM) \cite{han2000mining, pei2007h} also has received much attention in recent years. However, FIM only concentrates on  quantity but ignores high unit profit items which are infrequent. For example, diamond always brings high profit but milk is far more cheaper than it, and people can drink milk every day but hardly need diamond in daily life. In this case, high profitable goods will be supposed as the infrequent and uninteresting items. In fact, market retailers are not only in favor of  small profits and quick returns, but also want to earn substantial profits in short period. Hence, FIM algorithms cannot be competent for measuring other important factors in data like risk, profit, or weight of items as well.

Inspired by the utility theory \cite{Hutchison1963}, a new framework called high-utility itemset mining (HUIM) \cite{gan2021survey} was proposed. The utility concept is fairly broad and it can represent unit profit, interest, risk or other useful factor. In order to make this paper easier to understand, we assume utility is to identify how much profit an item/itemset can bring or make for users in the following content. Similar to the support threshold works in FIM, an user-specified threshold named minimum utility (abbreviated as \textit{minUtil}) is used to filter out unpromising items/itemsets. If the real utility of an item/itemset is no less than \textit{minUtil}, it will be supposed as promising since it is a high-utility pattern. In practice, because of utility metric, the mining result of HUIM is more interpretable than that of FIM. Therefore, in the past decades, HUIM algorithms have been further studied to discover valuable knowledge in different applications, such as user behavior analysis \cite{shie2013mining}, website click-stream analysis \cite{chu2008efficient}, and cross-marketing analysis \cite{yen2007mining}. Nevertheless, HUIM is a more challenge task than FIM, because the \textit{downward-closure} property \cite{agrawal1994fast} of FIM does not hold in HUIM, which means that the supersets of a low-utility itemset may be a subset of a high-utility itemset (abbreviated as HUI) incidentally. Hence, Liu \textit{et al.} \cite{liu2005two} proposed an overestimation concept named transaction-weighted utilization (abbreviated as \textit{TWU}) to address this issue. If \textit{TWU} of an itemset is less than the user-specified \textit{minUtil}, it would be an unpromising itemset; otherwise, it maybe a real HUI which needs to be confirmed. After that, there are many studies on improving performance of mining HUIs with different data structures, such as list-based algorithms (e.g., HUI-Miner \cite{liu2012mining} and FHM \cite{fournier2014fhm}), tree-based algorithms (e.g., UP-Growth \cite{tseng2012efficient} and MU-Growth \cite{yun2014high}) and projection-based algorithms (e.g., EFIM \cite{zida2017efim}, and TOPIC \cite{chen2021topic}).

However, a discovered HUI only provides information about its utility and the consisted items for decision makers. In fact, it hardly analyzes other useful information found by HUIM algorithms, such as quantity interval of each item in HUI. For example, sometimes, decision makers want to learn about ``the class of beautiful women" or ``the class of tall men" from the quantitative databases. The two adjectives ``beautiful" and ``tall" are both linguistic terms which cannot be directly described by numerical value. What's more, the fact remains that such inaccurate defined ``classes" plays a pivotal role in human thinking, especially in the explainable artificial intelligence system \cite{dovsilovic2018explainable}, pattern recognition and communication of information. Thus, fuzzy set theory \cite{zadeh1965fuzzy} which is simplicity and comprehensibility has been widely studied. Wang \textit{et al.} \cite{wang2009fuzzy} firstly proposed a new task named fuzzy utility mining (FUM). They integrated the fuzzy set theory and high-utility pattern mining algorithm to discover high fuzzy utility itemset (abbreviated as HFUI) from quantitative transaction databases, but it still does not keep the \textit{downward-closure} property. Due to above issues, Lan \textit{et al.} \cite{lan2015fuzzy} proposed a new fuzzy utility algorithm named two-phase fuzzy utility mining (abbreviated as TPFU), which considered external utility (i.e., unit profit) and internal utility (i.e., quantity) of items and the minimum operator principle of fuzzy set theory, to find out a complete set of real HFUIs in quantitative database. Since the \textit{downward-closure} property in fuzzy utility mining cannot be kept, they proposed an efficient fuzzy utility upper-bound model (simplified as FUUB) to solve this issue. With the fuzzy set theory, a set of high fuzzy utility itemsets is a class of objects with a continuum of grades by membership function. For example, given a quantitative itemset $\{milk(10), bread(30)\}$, in which the numbers represent the purchase quantities of corresponding goods. And we assume the unit profits of milk is \$1 and bread is \$6 respectively. The membership function is consist of three fuzzy regions: \textit{High}, \textit{Middle} and \textit{Low}. Then, the occurred quantities of two items can be converted into two different fuzzy sets $f_{milk}$ = $\{$1/\textit{milk.Low}, 0/\textit{milk.Middle}, 0/\textit{milk.High}$\}$ and $f_{bread}$ = $\{$0.6/\textit{bread.Low}, 0.4/\textit{bread.Middle}, 0/\textit{bread.High}$\}$. Then we can get two distinct fuzzy itemsets $\{$\textit{milk.Low}, \textit{bread.Low}$\}$ and $\{$\textit{milk.Low}, \textit{bread.Middle}$\}$. Take the first fuzzy itemset as a sample, the membership values of two fuzzy items $\{$\textit{milk.Low}, \textit{bread.Low}$\}$ are 1 and 0.6. The fuzzy utility of itemset $\{$\textit{milk.Low}, \textit{bread.Middle}$\}$ can be calculated as 0.6 $\times$ ((10 $\times$ \$1) + (30 $\times$ \$6)) = \$114 by the minimal operation. Compared with the original utility (= \$190), fuzzy utility value shows the combination of milk and bread is not very welcomed by local customers.

To summarize, the state-of-the-art TPFU is an Apriori-like algorithm in fact, although it proposed FUUB to reduce the search space during the level-wise manner. Similar to the Apriori algorithm, this two-phase model still suffers from 1) generating a huge number of candidates; 2) scanning database repeatedly; and 3) consuming too much runtime and memory. In light of these challenges, we propose the remaining fuzzy utility concept and design an efficient algorithm named \textbf{F}uzzy \textbf{U}tility-driven \textbf{I}temset \textbf{M}iner (FUIM) in this paper. The novel algorithm is proposed to efficiently identify high fuzzy utility itemsets (HFUIs). At last, experimental results show that FUIM has a good performance in terms of execution efficiency under various parameter settings. The key contributions of this work can be summarized as follows.

\begin{itemize}
	\item  We utilize the \textit{downward-closure} property of fuzzy utility itemset mining to find out HFUIs, and reduce the resource consumption effectively.
	
	\item  We design a novel data structure called fuzzy-list. It not only compresses the whole fuzzy utility message about a fuzzy itemset, but also provides necessary information for whether the fuzzy itemset should be cut off or not.
	
	\item  We firstly propose the remaining fuzzy itemset, which help extend low level HFUIs to high level HFUIs. And its corresponding remaining fuzzy utility is used to calculate a tight upper bound, which can substantially reduce the search space and memory consumption.
	
	\item  We test enough experiments on popular benchmarks, including synthetic and real datasets, demonstrating our novel algorithm can effectively discover complete set of HFUIs in detail.
	
\end{itemize}

The following content of this paper is organized as follows. In Section \ref{sec:relatedWork}, we briefly review the related work. The preliminaries and problem statement are presented in Section \ref{sec:preliminaries}. The proposed FUIM algorithm with detailed data structure and upper-bounds are described in Section \ref{sec:algorithm}. Several experimental evaluations verifying the efficiency and effectiveness of the FUIM approach are shown in Section \ref{sec:experimental}. Conclusion and future work are finally presented in Section \ref{sec:conclusion}.

\section{Related Work}
\label{sec:relatedWork}

This section briefly reviews related studies about high-utility itemset mining and fuzzy utility mining.

\subsection{High-Utility Itemset Mining}

In high-utility itemset mining (HUIM) \cite{gan2021survey}, both quantities and unit utility of items are considered. Compared with Apriori \cite{agrawal1994fast} (a level-wise association rule mining algorithm), HUIM aims to discover high-utility itemsets which are more explainable and meaningful than frequent ones to users. Due to its important practical applications, HUIM has gradually become an emerging research task in last decades. In 2003, Chan \textit{et al.} \cite{chan2003mining} introduced a idea about utility mining field, and Yao \textit{et al.} \cite{yao2006unified} then proposed a strict unified framework. However, as we described in previous content, the most difficult challenge of utility mining is the low-utility items may be contained in a high-utility itemset (HUI), which the \textit{download-closure} property belongs to frequent itemset mining does not hold in. After that, Liu \textit{et al.} \cite{liu2005two} proposed a novel model using transaction-weighted utilization (\textit{TWU}) to discover HUIs efficiently. Based on the developed transaction-weighted downward closure (TWDC) property, it confirms that a HUI is impossible including any uninteresting itemsets. Their algorithm is updated based on the Apriori method, which follows the generate-and-test mechanism. Therefore, the multiply execution times for scanning database is inevitable, and resource consumption is unacceptable too.

Except the previous algorithms we have introduced, IHUP \cite{ahmed2009efficient} discovered incremental and interactive HUIs based on tree structure well. At the same time, Tseng \textit{et al.} proposed two tree-based algorithms namely UP-growth \cite{tseng2010up} and UP-growth+ \cite{tseng2012efficient}. With the compact utility pattern tree and several efficient pruning strategies, it had deeply reduced the overestimated utility values and enhance the performance about mining. Recently, Liu \textit{et al.} \cite{liu2012mining} proposed a novel data structure named utility list and the HUI-Miner algorithm. In fact, HUI-Miner performs better than IHUP because it can effectively find out complete HUIs without generating candidates. Each item/itemset owns a unique compact utility list structure, and the utility-list actually is consisted of a tuple of three terms: 1) transactions' ID that represents which transaction contain this item/itemset; 2) the real utility of item/itemset in these transactions; and 3) the utility of extended pattern of item/itemset in these transactions. Due to all key information is compressed in these utility-lists, it is no need to generate candidates in advance and then filter out low-utility itemsets. Experiments also show that HUI-Miner considerably outperforms the previous algorithms for mining HUIs \cite{liu2012mining}. Later, inspired by HUI-Miner, FHM \cite{fournier2014fhm} adopted the same data structure. Furthermore, it proposed an Estimated Utility Co-occurrence Structure to discover HUIs. In recent years, Zida \textit{et al.} \cite{zida2017efim} adopted pseudo-projection technology to reduce the search space significantly, which outperforms HUI-Miner and FHM in most databases. In addition, researchers also developed variations of the problem of HUIM such as top-$k$ HUIM \cite{gan2020tophui,tseng2015efficient}, on-shelf HUIM \cite{chen2020osumi,lan2011discovery}, discovering HUIs in dynamic environment \cite{gan2018survey,lin2015fast}, finding out the concise and lossless representation of HUI \cite{tseng2014efficient}, and mining the up-to-date HUIs \cite{lin2015efficient}. Besides, it should be pointed out that all the above algorithms we listed rely on the \textit{TWU} model. If \textit{TWU} value of an itemset is less than \textit{minUtil}, then we can directly suppose this itemset is not a potential HUI. On the contrary, it needs to compute its real utility to determine whether is a HUI or not. Many other advanced developments for utility mining can be referred to literature review \cite{gan2021survey}.

\subsection{Fuzzy Utility Mining}

Thanks to the comprehensibility and simplicity attributes, the fuzzy set theory \cite{zadeh1965fuzzy} has been widely used in various intelligent systems or other applications to improve flexibility of decision making \cite{lin2010linguistic}. In generally, membership function will compute fuzzy values of distinct patterns within an interval [0, 1], and the value indicates the membership level of pattern in different classes. According to the membership level, decision makers will understand more accurately customers' preference. In 1998, Kuok \textit{et al.} \cite{kuok1998mining} proposed a new research task named fuzzy data mining, which had successfully integrated the fuzzy set theory with data mining techniques. According to fuzzy set theory, the trouble introduced in Section \ref{sec:introduction} is easy to solve. The quantity of each item in transactions will be converted into linguistic regions, which are more suitable for human's mind. Compare with fuzzy data mining domain, traditional association rules \cite{agrawal1994fast} and quantitative rules \cite{chan1997mining, hong1999mining} offer interesting patterns without quantitative knowledge like linguistic regions \cite{hong2016survey}. In traditional HUIM field, an HUI provides the contained items and its total utility information. HUIM is so intuitive and has been applied in few cases. In real life applications, however, some one like eating a piece of bread but some prefer three pieces. Because the traditional HUIM algorithms hide the quantity information in results, decision makers does not get the original quantity information. The high-utility itemset only offers message that bread is a profitable good, therefore, a suitable promotion quantity of pattern is quite important. However, due to a traditional item/itemset will be converted into some fuzzy items/itemsets (depend on the membership function). Actually, the fuzzy utility mining task is more complicated than HUIM.

Later, Wang \textit{et al.} \cite{wang2009fuzzy} combined fuzzy set concept with HUIM technique to discover high fuzzy utility itemsets (abbreviated as HFUIs). Specifically, they defined a membership function which maps quantitative attribute of an item/itemset to the corresponding linguistic region value, and then calculated the fuzzy utility to find HFUIs. What's more, as mentioned previously, fuzzy theory adopts a minimal operator to evaluate the overlap value of linguistic region value of distinct items/itemsets. One of the biggest problems of \cite{wang2009fuzzy} is ignoring to obtain the common degree values of fuzzy itemsets. After that, Lan \textit{et al.} \cite{lan2015fuzzy} introduced a new upper-bound model for fuzzy utility mining to improve performance. Since their method is still Apriori-like method which may easily suffer from the costly computation and memory. Chen \textit{et al.} \cite{chen2014actionable} introduced the concept of actionable high-coherent-utility fuzzy itemset. Besides, several methods for temporal-based fuzzy utility mining \cite{hong2020one,huang2017temporal}, like level-wise methods and tree-based methods, are proposed to deal with temporal data. All in all, it motivates us to continue to explore this issue. In this paper, we develop an efficient fuzzy utility mining algorithm which performs better than the state-of-the-art TPFU algorithm.

\section{Preliminaries and Problem Formulation}
\label{sec:preliminaries}

In this section, the frequently used notations in this paper are firstly given in Table \ref{tab:table_Notation}. We also adopt some definitions from previous studies \cite{lan2015fuzzy,chen2014actionable} for clear expression of the research issue. In addition, Table \ref{tab:database} is a simple quantitative database. It is consisted of ten transactions with five fuzzy items \{$A$, $B$, $C$, $D$ and $E$\}, and the corresponding unit utility of each item is \$2, \$6, \$3, \$8, and \$10, respectively. Finally, the problem definition about fuzzy-driven utility itemset mining (FUIM) is formalized.

\begin{table}[!htbp]
	\centering
	\small
	\caption{Summary of notations}
	\label{tab:table_Notation}
	\begin{tabular}{|c|l|}
		\hline
		\textbf{Symbol} & \makebox[6cm][c]{\textbf{Description}}  \\ \hline
		$I$ &  A set of $m$ items, \textit{I} = \{\textit{x}$_{1}$, \textit{x}$_{2}$, $\ldots$, \textit{x}$_{m}$\}. \\ \hline		
		$X$ & An itemset $X$ = $\{x_1$, $x_2$, $ \ldots, x_i\}$. \\ \hline		
		$\mathcal{D}$	&  A quantitative database, $\mathcal{D}$ = \{\textit{T}$_{1}$, \textit{T}$_{2}$, $\ldots$, \textit{T$_{j}$}\}.  \\ \hline
		
		$\gamma$ &   A  minimum fuzzy utility threshold. \\ \hline	
		
		$q(x_i, T_j)$ &   The occurred quantity of an item $x_i$ in $T_j$. \\ \hline
		$eu(x_i)$ &   Each item $x_i \in I$ has an external utility. \\ \hline
		
		\textit{fu} &  The fuzzy utility value. \\ \hline
		\textit{mfu} &  The maximal fuzzy utility value. \\ \hline
		\textit{mtfu} &  The maximal transaction fuzzy utility value. \\ \hline
		\textit{rfu} &  The remaining fuzzy utility value. \\ \hline
		\textit{HFUUBI}  &  The high fuzzy utility upper-bound itemset. \\ \hline	
		\textit{HFUI} &   The high fuzzy utility itemset. \\ \hline
	\end{tabular}
\end{table}

\subsection{Fuzzy-based Concepts}

\begin{definition}
	\rm The fuzzy set $f_{ij}$ of the quantitative value $q(x_i, T_j)$ of the $i$-th item $x_i$ in the $j$-th transaction $T_j$ can be denoted as $f_{ij}$ = $\big($$\frac{f_{ij1}}{R_{i1}}$ + \ldots + $\frac{f_{ijl}}{R_{il}}$ + \ldots + $\frac{f_{ijm}}{R_{im}}$$\big)$ with the given membership function, where $m$ is the number of regions of the item $x_i$. $R_{il}$ is the $l$-th fuzzy region of $x_i$, and $f_{ijl}$ is the fuzzy membership value of the $q(x_i, T_j)$ of the $i$-th item $x_i$ in the $l$-th fuzzy region $R_{il}$, where $f_{ijl}$ $\in$ [0, 1].
\end{definition}

For example, with Table \ref{tab:database} and Fig. \ref{fig:Membership}, the quantitative value (= 3) of item $A$ in $T_2$ can be converted to $f_{A, 2}$ = (0.3/\textit{A.Low}, 0.7/\textit{A.Middle}, 0/\textit{A.High}) by the membership function.

\begin{figure}[hbt]
	\centering
	\includegraphics[scale=0.48]{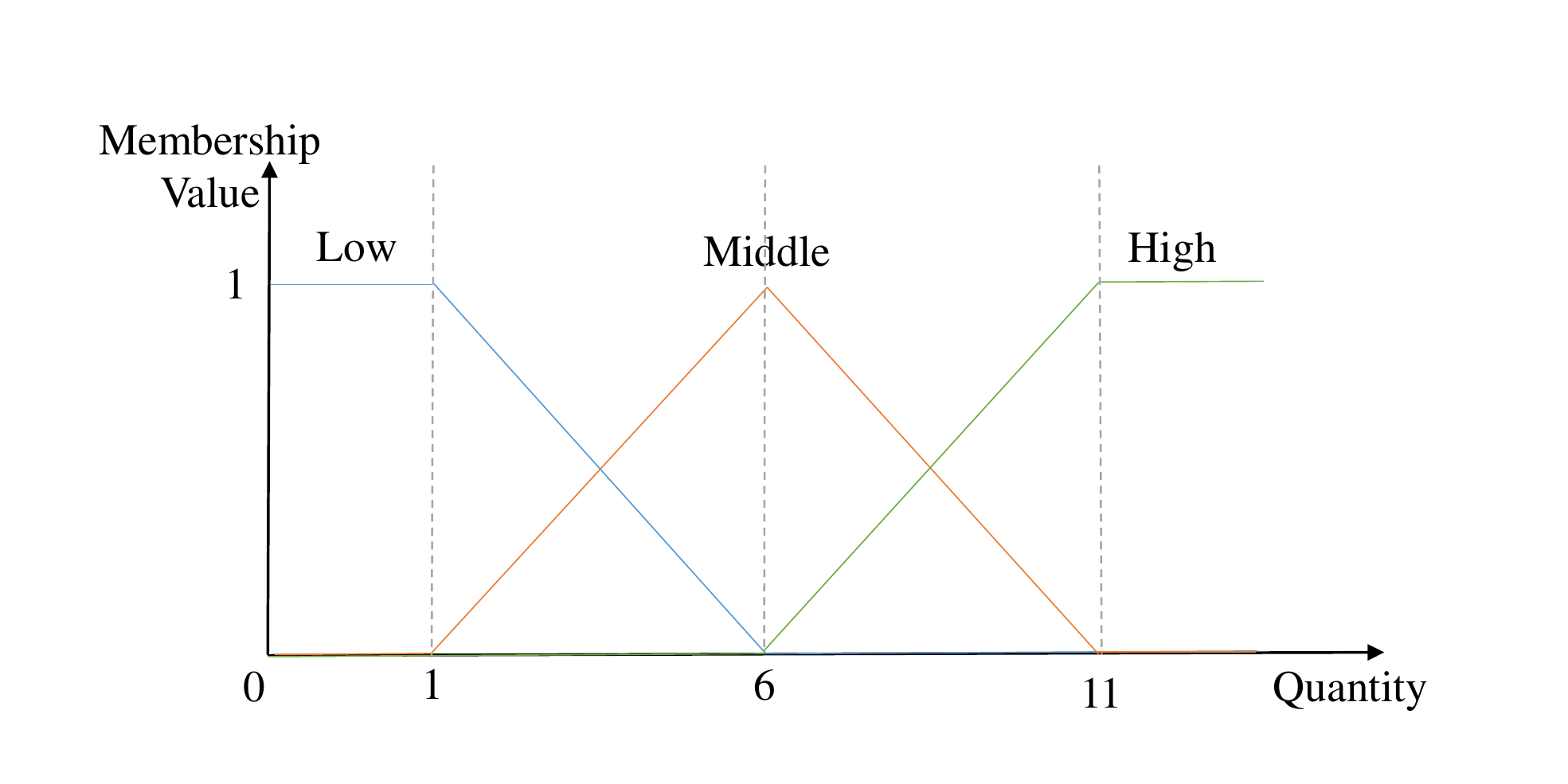}
	\caption{The membership function.}
	\label{fig:Membership}
\end{figure}

\begin{table}[!h]
	\begin{center}
		\caption{A sample quantitative database}
		\label{tab:database}
		\begin{tabular}{lccccc}
			\hline
			\textbf{TID} & $A$ & $B$ & $C$ & $D$ & $E$ \\ \hline
			$T_1$         & 0 & 2 & 5 & 2 & 0 \\
			$T_2$         & 4 & 6 & 0 & 0 & 0 \\
			$T_3$         & 0 & 3 & 6 & 0 & 4 \\
			$T_4$         & 2 & 2 & 7 & 0 & 0 \\
			$T_5$         & 2 & 0 & 8 & 0 & 0 \\
			$T_6$         & 6 & 5 & 0 & 4 & 0 \\
			$T_7$         & 4 & 4 & 0 & 7 & 3 \\
			$T_8$         & 0 & 2 & 3 & 0 & 0 \\
			$T_9$         & 0 & 0 & 0 & 3 & 3 \\
			$T_{10}$      & 0 & 0 & 0 & 2 & 0 \\ \hline
		\end{tabular}
	\end{center}
\end{table}

\begin{definition}
	\rm  The fuzzy utility of the $l$-th fuzzy region of an item $x_i$ in a transaction $T_j$ is defined as: \textit{fu}$_{ijl}(x_i, T_j)$ = $f_{ijl}$ $\times $ $q(x_i, T_j)$ $\times$ \textit{eu}$(x_i)$. Furthermore, the utility of a fuzzy itemset $X$ in $T_j$ is denoted as \textit{fu}$_{jX}$($X$, $T_j$) = $f_{jX}$ $\times$ $\sum_{x_i \subseteq X \land X \subseteq T_j}$($q$($x_i$, $T_j$) $\times$ \textit{eu}($x_i$)), where $f_{jX}$ is the minimal membership value of all items in $X$.
\end{definition}

For example in Table \ref{tab:database}, given a fuzzy itemset $X$ = $\{$\textit{A.Low}, \textit{C.Middle}$\}$ in transaction $T_5$, the membership values of two fuzzy items are 0.3 and 0.6, respectively, based on Fig. \ref{fig:Membership}. According to the minimum operation mechanism, the membership value of 2-itemset $\{$\textit{A.Low}, \textit{C.Middle}$\}$ is 0.3 in $T_5$. Thus, the fuzzy utility of 2-itemset is $fu_{\{5, X\}}(X, T_5)$ = 0.3 $\times$ ((2 $\times$ \$2) + (8 $\times$ \$3)) = \$8.4.

\begin{definition}
	\rm The total fuzzy utility $fu_{il}$ of $l$-th fuzzy region of an item $x_i$ in $\mathcal{D}$ is the summation of fuzzy utility values of all occurrence $x_i$, which is formulated as \textit{fu}$_{il}$($x_i$) = $\sum_{x_i \in T_j \land T_j \subseteq \mathcal{D}}$\textit{fu}$_{ijl}$($x_i$, $T_j$). The total fuzzy utility \textit{fu}$_{X}$ of $X$ takes the same way, whereby \textit{fu}$_{X}$ = $\sum_{X \subseteq T_j \land T_j \subseteq \mathcal{D}}$\textit{fu}$_{jX}$. In addition, in a quantitative database, the \textit{fu}$_{jX}$ is the fuzzy utility of $X$ in the $j$-th transaction.
\end{definition}

For example, the total fuzzy utility of item \textit{E.Low} in $\mathcal{D}$ is $fu_{\{E, 3, Low\}}$ + $fu_{\{E, 7, Low\}}$ + $fu_{\{E, 9, Low\}}$ = \$43.3, and the fuzzy itemset $X$ = $\{$\textit{B.Low}, \textit{C.Middle}$\}$ is $fu_{\{1, X\}}$ + $fu_{\{3, X\}}$ + $fu_{\{4, X\}}$ + $fu_{\{8, X\}}$ = \$78 from Table \ref{tab:database}.

\begin{definition}
	\rm If the fuzzy utility of a fuzzy itemset $X$ is no less than user-specified minimum fuzzy utility threshold $\gamma$ ($fu_{X} \ge \gamma$), then $X$ is a high fuzzy utility itemset (abbreviated as HFUI).
\end{definition}

For example, if set $\gamma$ = \$60\footnote{If not particularly indicated, the following content will always regard $\gamma$ as \$60.}, the fuzzy 2-itemset $\{$\textit{B.Low}, \textit{C.Middle}$\}$ is a HFUI. On the contrary, consider the fuzzy utility of another 2-itemset $\{$\textit{A.Low}, \textit{C.Middle}$\}$ is \$36.8, which is less than \$60. Clearly, it is a low fuzzy utility 2-itemset. The other HFUIs are shown in Table \ref{tab:HFUIs}. 

\begin{table}[!h]
	\begin{center}
		\caption{The high fuzzy utility itemset with $\gamma$ = \$60}
		\label{tab:HFUIs}
		\begin{tabular}{|l|l|}
			\hline
			\makebox[4cm][c]{\textbf{Fuzzy itemset}} & \textbf{Utility} \\ \hline
			$\{$\textit{C.Middle}, \textit{B.Low}$\}$ & \$78 \\ \hline
			$\{$\textit{E.Low}, \textit{D.Middle}$\}$ & \$73.2  \\ \hline
			$\{$\textit{E.Low}, \textit{D.Middle}, \textit{B.Middle}$\}$ & \$66 \\ \hline
			$\{$\textit{E.Low}, \textit{D.Middle}, \textit{A.Middle}, \textit{B.Middle}$\}$ & \$70.8 \\ \hline
			$\{$\textit{D.Middle}$\}$ & \$80 \\ \hline
			$\{$\textit{D.Middle}, \textit{A.Middle}$\}$ & \$64.8 \\ \hline
			$\{$\textit{D.Middle}, \textit{A.Middle}, \textit{B.Middle}$\}$ & \$97.2 \\ \hline
			$\{$\textit{D.Middle}, \textit{B.Middle}$\}$ & \$90.8  \\ \hline
			$\{$\textit{A.Middle}, \textit{B.Middle}$\}$ & \$82.4 \\ \hline
			$\{$\textit{B.Middle}$\}$ & \$88.8  \\ \hline
			$\{$\textit{C.Middle}$\}$ & \$64.8 \\
			\hline
		\end{tabular}
	\end{center}
\end{table}

\subsection{Problem Formulation}

Based on the definitions we introduced above, the problem of fuzzy-driven utility itemset mining is formulated below.

\textbf{Problem statement}. Given a quantitative transaction database $\mathcal{D}$, a user-specified minimum utility threshold $\gamma$, and a user-defined membership function, the goal of fuzzy-driven utility itemset mining is to identify a complete set of HFUIs. Our paper aims to find the itemsets whose fuzzy utilities are no less than $\gamma$ in $\mathcal{D}$.

\section{Design Algorithm}
\label{sec:algorithm}

\subsection{Fuzzy Utility Upper Bound}

As shown in Table \ref{tab:HFUIs}, \{\textit{B.Middle}\} is a high fuzzy utility item but \{\textit{A.Middle}\} is not. However, their superset \{\textit{A.Middle}, \textit{B.Middle}\} is a HFUI because of $fu_{\{A.Middle, B,Middle\}}$ $>$ $\gamma$. This represents that fuzzy utility concept does not hold \textit{downward-closure} property of traditional ARM algorithms, which indicates the fuzzy utility mining task is more difficult than frequent itemset mining. To resolve this issue, we adopt an effective fuzzy utility upper-bound model (abbreviated as FUUB) from Ref. \cite{lan2015fuzzy}. The details about related terms are introduced in the following.

\begin{definition}
	\rm  In transaction $T_j$, the maximum fuzzy utility of a fuzzy item $x_i$ is denoted as \textit{mfu}$_{ij}$ = \textit{max}\{$fu_{ij1}$, $fu_{ij2}$, $\ldots$, $fu_{ijl}$\}, where $fu_{ijl}$ represents the fuzzy utility of the $l$-th fuzzy region of $x_i$ in $T_j$.
\end{definition}

For example, Fig. \ref{fig:Membership} and Table \ref{tab:database} show the membership values of item $B$ in transaction $T_1$ are 0.6/\textit{B.Low}, 0.4/\textit{B.Middle}, and 0/\textit{B.High}, respectively. The unit utility of $B$ is \$6, thus $fu_{\{B.Low, 1\}}$ = 0.6 $\times$ 2 $\times$ \$6 = \$7.2, $fu_{\{B.Middle, 1\}}$ = \$4.8 and $fu_{\{B.High, 1\}}$ = \$0. Obviously, \textit{mfu}$_{\{B, 1\}}$ is \$7.2 with the definition of the maximum fuzzy utility.

\begin{definition}
	\rm The maximum transaction fuzzy utility of an transaction $T_j$ in $\mathcal{D}$ is denoted as \textit{mtfu}$_j$, where \textit{mtfu}$_j$ = $\sum_{x_i \in T_j}$\textit{mfu}$_{ij}$ and \textit{mfu}$_{ij}$ is the maximum fuzzy utility of the $l$-th fuzzy region of $x_i$ in $T_j$ \cite{lan2015fuzzy}.
\end{definition}

For example, in transaction $T_2$, the membership values of two items $A$ an $B$ are $\{$0.4/\textit{A.Low}, 0.6/\textit{A.Middle}, 0/\textit{A.High}$\}$ and $\{$0/\textit{B.Low}, 1/\textit{B.Middle}, 0/\textit{B.High}$\}$, respectively. Therefore, \textit{mtfu}$_2$ = \textit{mfu}$_{\{A.Middle, 2\}}$ + \textit{mfu}$_{\{B.Middle, 2\}}$ = \$40.8.

\begin{definition}
	\rm  Based on definition about the maximum transaction fuzzy utility of transaction, the FUUB of a fuzzy itemset $X$ is formulated as \textit{fuub}($X$) = $\sum_{X \subseteq T_j \land T_j \subseteq \mathcal{D}}$\textit{mtfu}$_j$ \cite{lan2015fuzzy}. In another word, the FUUB of $X$ is the summation of fuzzy utility of transactions containing $X$. As similarly as the definition of HFUI, if \textit{fuub}($X$) is no less than threshold $\gamma$, then we suppose $X$ is a potential pattern called high fuzzy utility upper-bound itemset (abbreviated as HFUUBI); otherwise it is a low fuzzy utility upper-bound itemset, which cannot be a HFUI.
\end{definition}

For example, the Table \ref{tab:hfuubi} lists a complete set of potential fuzzy 1-itemsets. \textit{fuub}($C$) = \$149.4 that $C$ is a HFUUBI, and then its real fuzzy utility is \$64.8 $>$ $\gamma$ from Table \ref{tab:HFUIs}.

\begin{table}[!h]
	\begin{center}
		\caption{The high fuzzy utility upper-bound of 1-itemset}
		\label{tab:hfuubi}
		\begin{tabular}{lccccc}
			\hline
			\textbf{FHUUBI} & $A$ & $B$ & $C$ & $D$ & $E$ \\ \hline
			\textbf{Utility} & \$225.2 & \$309.8 & \$149.4 & \$216.8 & \$167.2 \\ \hline
		\end{tabular}
	\end{center}
\end{table}

\begin{theorem}
	\textbf{(The fuzzy-utilization downward closure property \cite{lan2015fuzzy})}
	\label{theo:fuzzyDC}
	\rm Let $I^{k-1}$ be a ($k$-1)-itemset and $I^k$ be a $k$-itemset where $I^{k-1} \subset I^k$. If $I^k$ is a HFUUBI, then $I^{k-1}$ must be a HFUUBI too.
\end{theorem}

\begin{proof}
	Let $T_{I^k}$ be a set of transactions containing $I^k$ and $T_{I^{k-1}}$ be a set of transactions including $I^{k-1}$. Since $I^{k-1} \subset I^{k}$, $\mid$$T_{I^k}$$\mid$ $\subset$ $\mid$$T_{I^{k-1}}$$\mid$. Based on the definition about FUUB, \textit{fuub}($I^{k-1}$) = $\sum_{I^{k-1} \subseteq T_j \land T_j \subseteq \mathcal{D}}$\textit{mtfu}$_j$ $\ge$ $\sum_{I^k \subseteq T_j \land T_j \subseteq \mathcal{D}}$\textit{mtfu}$_j$ = \textit{fuub}($I^k$).
\end{proof}

\begin{theorem}
	\textbf{(The fuzzy utility upper-bound constrain \cite{lan2015fuzzy})}
	\label{theo:upperbound}
	\rm As introduced in previous definitions, HFUIs is the collection of high fuzzy utility itemsets in quantitative database $\mathcal{D}$ and HFUUBIs is the collection of all potential high fuzzy utility itemsets in $\mathcal{D}$. Therefore, HFUIs is a subset of HFUUBIs (HFUIs $\subseteq$ HFUUBIs). This means that if a fuzzy itemset with low fuzzy utility upper-bound, its real fuzzy utility must be less than $\gamma$.
\end{theorem}

\begin{proof}
	For $\forall x_i \in X$ $\land X \subseteq $ $T_j \subseteq \mathcal{D}$, then: \\
	$\because$ \textit{mfu}$_{ij}$ = \textit{max}$\{$\textit{fu}$_{ij1}$, \ldots, \textit{fu}$_{ijl}$$\}$
	\begin{tabbing}
		$\therefore$ $fuub(X)$ \= $=\sum_{X \subseteq T_j \land T_j \subseteq \mathcal{D}}$\textit{mtfu}$_{ij}$ \\
		\>$=\sum_{X \subseteq T_j \land T_j \subseteq \mathcal{D}}$$\sum_{x_i \in X}$\textit{mfu}$_{ij}$ \\
		\>$\ge\sum_{X \subseteq T_j \land T_j \subseteq \mathcal{D}}$$\sum_{x_i \in X}$\textit{fu}$_{ijl}$ \\
		\>$=\sum_{X \subseteq T_j \land T_j \subseteq \mathcal{D}}$\textit{fu}$_{jX}$ \\
		\>$=fu_{X}$.
	\end{tabbing}
	Therefore, if \textit{fubb}($X$) $<$ $\gamma$ that \textit{fu}$_{X}$ $<$ $\gamma$.
\end{proof}

\subsection{Fuzzy-list Structure}

\rm Before introduce the new fuzzy data structure, we firstly claim the order adopted in our novel algorithm. Let $\succ$ be a global order on fuzzy items from $I$, then an order sequence ``$C$ $\succ$ $E$ $\succ$ $D$ $\succ$ $A$ $\succ$ $B$" based on fuzzy utility upper-bound is produced. Furthermore, the transaction quantitative database is considered as revised too. Given a fuzzy itemset $X$, the set of all items after $X$ in $T_j$ is denoted as $T_j / X$.

\begin{definition}
	\rm  As shown in Fig. \ref{fig:liststructure_1}, we keep each fuzzy itemset $X$ with three elements that make up the fuzzy-list structure (simply named \textit{ful}): the transaction identifier (\textit{tid}), the internal fuzzy utility (\textit{ifu}), and the remaining fuzzy utility (\textit{rfu}). The term of \textit{tid} is transactions contain $X$. The term of \textit{ifu} represents the real fuzzy utility of $X$ in the transaction $T_j$, is defined as \textit{ifu}($X$, $T_j$) = \textit{fu}$_{jX}$($X$, $T_j$). The term of \textit{rfu} means the remaining maximal fuzzy utility value obtained by union operating, is denoted as \textit{rfu}($X$, $T_j$) = $\sum_{x_i \in T_j / X}$\textit{mfu}$_{ij}$.
\end{definition}

\begin{figure*}[hbtp]
	\centering
	\includegraphics[scale=0.56]{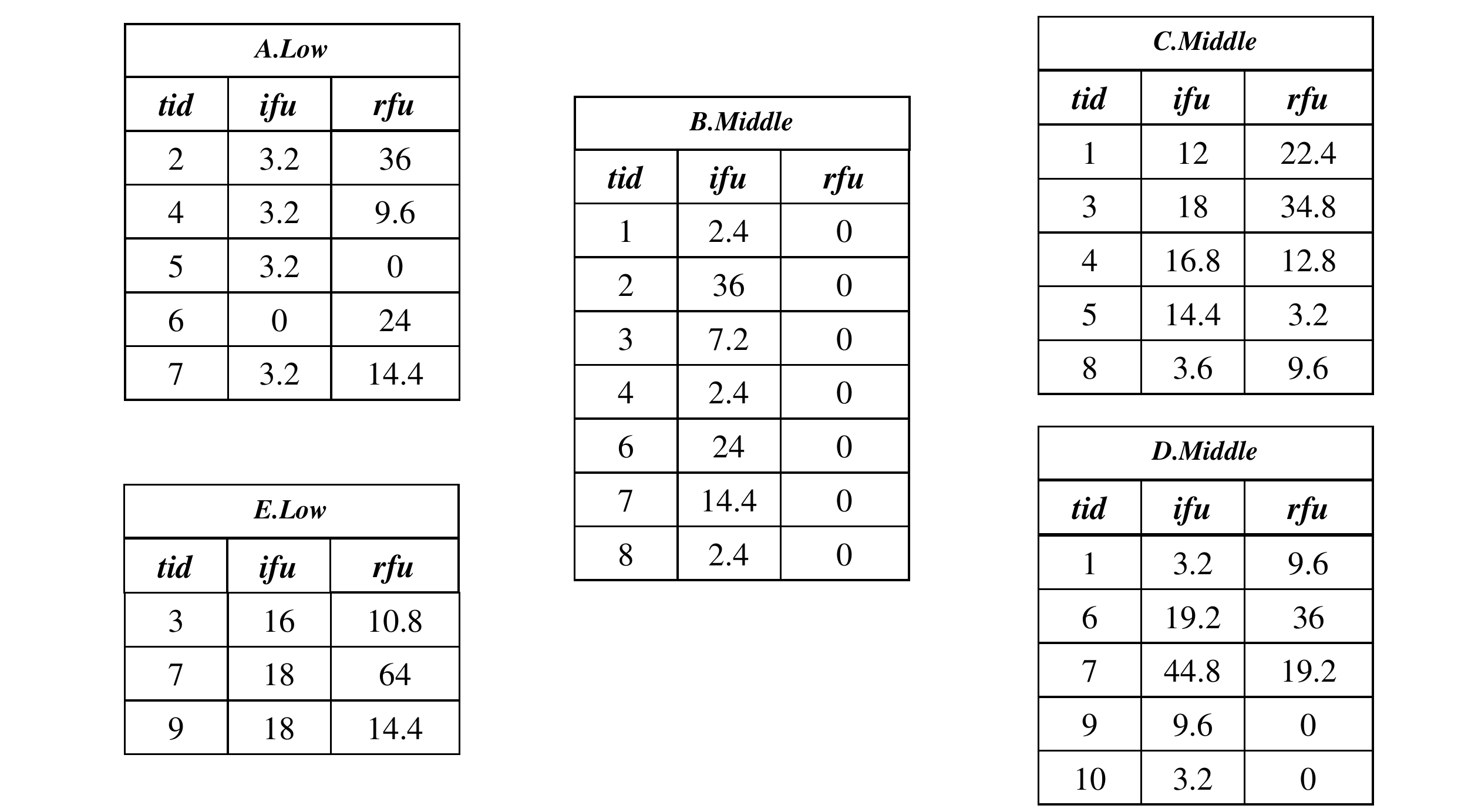}
	\caption{The constructed fuzzy-list structures of 1-itemsets.}
	\label{fig:liststructure_1}
\end{figure*}

Without database scanning, we can join fuzzy-lists of two distinct ($k$-1)-itemsets to form a new fuzzy-list structure of $k$-itemsets ($k \ge 2$). During the intersecting process, the tuples with the same \textit{tid} will be combined together. In order to accelerate the process, assume each row of list is with \textit{tid} ascending order, which can use binary search method to locate target tuple. If we suppose the size of two distinct lists are $m$ and $n$ respectively, the time complexity is $O$($m$ + $n$) in the worst, because all \textit{tids} in list are ordered. At the same time, in Fig. \ref{fig:liststructure_2}, there are three columns in fuzzy-list. The \textit{ifu} (the second column) indicates the real fuzzy utility of item/itemset in corresponding transaction and is easy to be calculated. Then the \textit{rfu} (the third column) represents the remaining fuzzy utility of item/itemset in corresponding transaction left.  As for \textit{rfu} of new generated fuzzy $k$-itemsets ($k \ge 2$), assume the minimal \textit{rfu} value of the merged ($k$-1)-itemsets as new remaining value. As shown in Fig. \ref{fig:liststructure_2}, the final eligible results about fuzzy 2-itemsets are listed as $\{$\textit{B.Middle}, \textit{A.Low}$\}$, $\{$\textit{E.Low}, \textit{B.Middle}$\}$, $\{$\textit{E.Low}, \textit{C.Middle}$\}$ and so on. Thus, the final fuzzy utility of $X$ is the summation of all internal fuzzy utility values, is denoted as \textit{sumIfu}($X$) = $\sum_{T_j \in ful(X)}$\textit{ifu}($X$, $T_j$). And remaining region terms are plan to fuzzy items in quantitative transaction $T_j$. Therefore, the remaining fuzzy utility of $X$ is the summation of the third column elements, is denoted as \textit{sumRfu}($X$) = $\sum_{x_i \in T_j / X \land T_j \subseteq ful(X)}$\textit{rfu}($x_i$, $T_j$).

For example, in Fig. \ref{fig:liststructure_2}, the \textit{sumIfu} of fuzzy itemset $\{$\textit{E.Low}, \textit{D.Middle}$\}$ is \$51.6 + \$21.6 = \$73.2. And \textit{sumRfu}(\textit{E.Low}) is computed as \$10.8 + \$64 + \$14.4 = \$89.2 in Fig. \ref{fig:liststructure_1}.

\begin{figure*}[hbtp]
	\centering
	\includegraphics[scale=0.59]{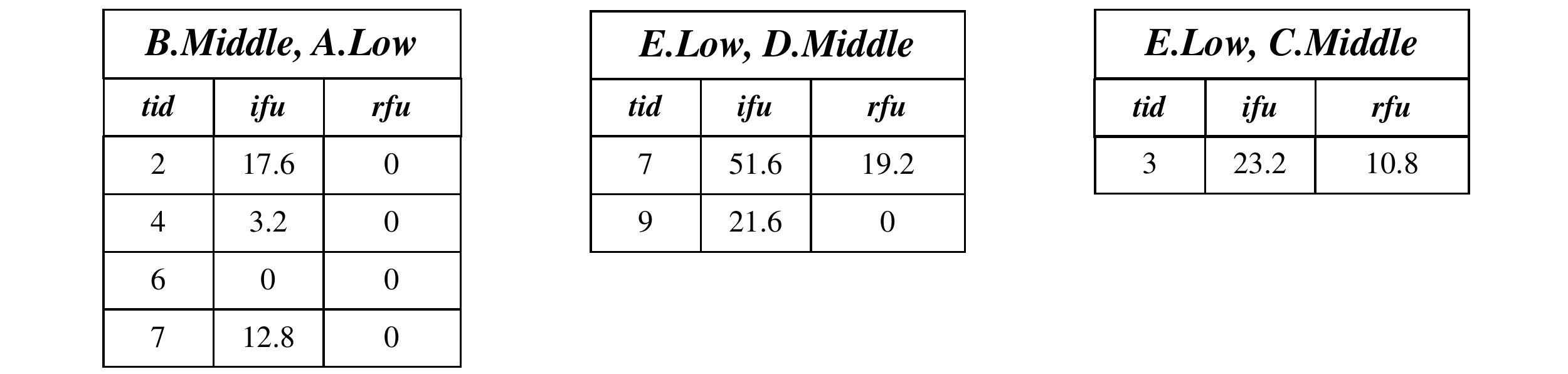}
	\caption{The fuzzy-list structures of fuzzy 2-itemsets.}
	\label{fig:liststructure_2}
\end{figure*}

\subsection{Searching Space}

\begin{figure}[hbtp]
	\centering
	\includegraphics[scale=0.45]{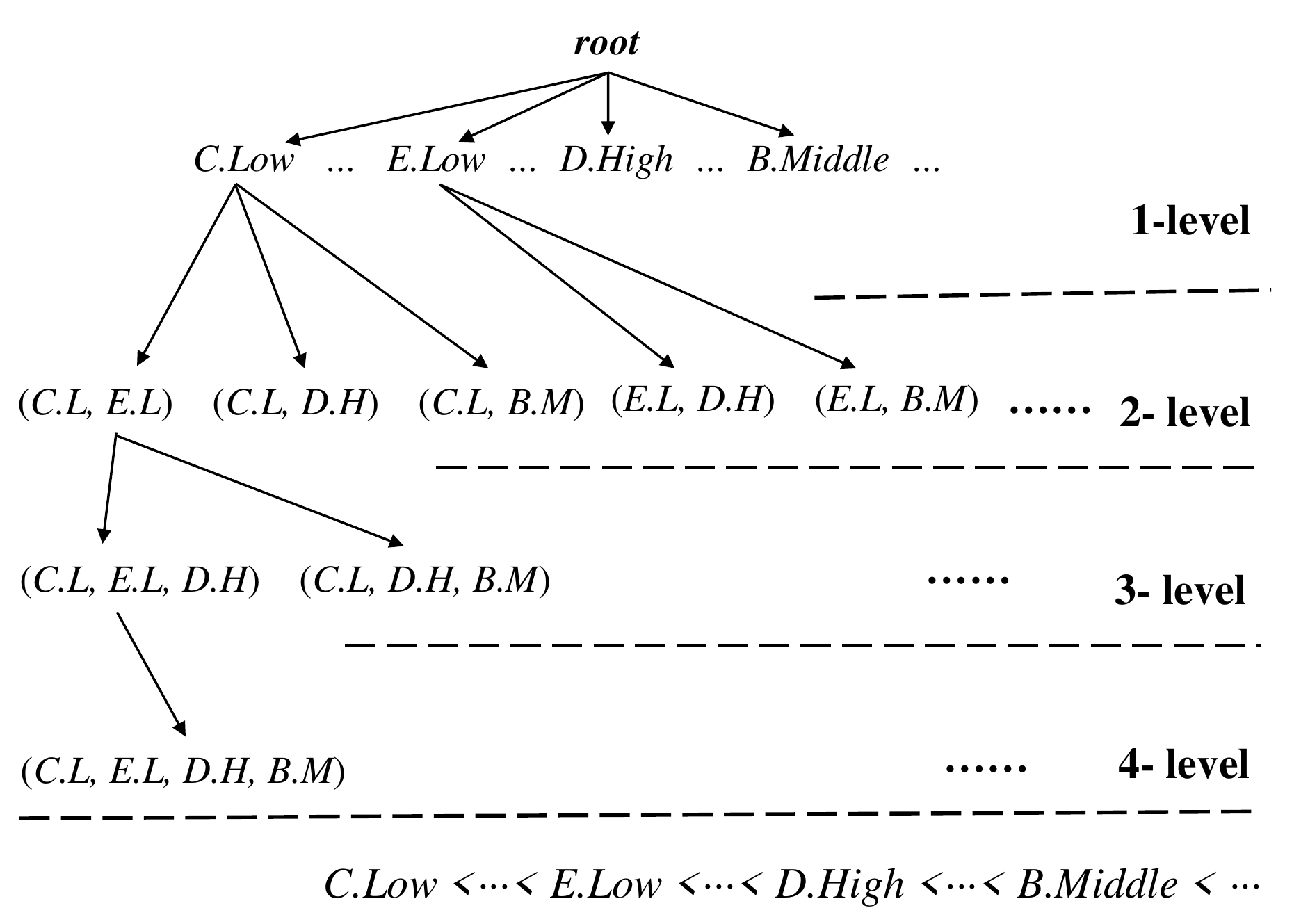}
	\caption{An enumeration tree of the used example.}
	\label{fig:enumeratetree}
\end{figure}

The searching space of fuzzy utility mining problem can be regard as a set-enumeration tree \cite{liu2005two}. In this paper, with a total order on all fuzzy items ($C$ $\succ$ $E$ $\succ$ $D$ $\succ$ $A$ $\succ$ $B$), a set-enumeration tree is depicted in Fig. \ref{fig:enumeratetree}. As we can see clearly, exhaustive search will be excessively time-consuming because of a huge number of nodes. If there is $n$ fuzzy items, it has to check $2^n$ itemsets in all. Thus, we adopt two pruning strategies to reduce the searching space. The specific description is as follows:

\begin{theorem}
	\textbf{(The internal fuzzy utility constrain)}
	\label{theo:sumIfu}
	\rm According to fuzzy-list structure, given a fuzzy itemset $X$, the internal fuzzy utility is its real fuzzy utility values, thus if the \textit{sumIfu}($X$) is no less than the minimum fuzzy utility threshold $\gamma$, it will be a high fuzzy utility itemset.
\end{theorem}

\begin{proof}
	Given a fuzzy itemset $X$ and its corresponding fuzzy-list \textit{ful}($X$), then:
	\begin{tabbing}
		$\because$ \textit{sumIfu}($X$) \=
		= $\sum_{X \subseteq T_j \land T_j \subseteq ful(X)}$\textit{ifu}($X$, $T_j$). \\
		\> = $\sum_{x_il \in X \land X \subseteq T_j \land T_j \subseteq ful(X)}$\textit{fu}($x_il$, $T_j$). \\
		$\therefore$ \textit{sumIfu}($X$) is the real fuzzy utility of $X$.
	\end{tabbing}
	Thus, if \textit{sumIfu}($X$) $\ge$ $\gamma$ that $X$ is a HFUI.
\end{proof}

\begin{theorem}
	\textbf{(The remaining fuzzy utility constrain)}
	\label{theo:sumRfu}
	\rm Given a fuzzy itemset $X$ and its fuzzy-list structure, if the sum of all the \textit{ifu} and \textit{rfu} is no less than $\gamma$, there exists some extension items of $X$ may be high fuzzy utility itemsets too. Otherwise, there is no need to reconstruct a new fuzzy-list structure and we can recall other nodes.
\end{theorem}

\begin{proof}
	Assume the extension of itemset $X$ as $X^\prime$, and for $\forall T_j \supseteq X^\prime$, thus: ($X^\prime - X$) = ($X^\prime / X$).\\
	$\because$ $X \subset X^\prime \subseteq T_j \Rightarrow (X^\prime / X) \subseteq (T_j / X)$.
	\begin{tabbing}
		$\therefore$ \textit{fu}$_{jX}$($X^\prime$, $T_j$) \=
		= \textit{fu}$_{jX}$($X$, $T_j$) + \textit{fu}$_{j(X^\prime - X)}$(($X^\prime - X$), $T_j$) \\
		\> = \textit{fu}$_{jX}$($X$, $T_j$) + \textit{fu}$_{j(X^\prime - X)}$(($X^\prime / X$), $T_j$) \\
		\> = \textit{fu}$_{jX}$($X$, $T_j$) + $\sum_{x_i \in (X^\prime / X)}$\textit{fu}$_{ijl}$($x_i$, $T_j$) \\
		\> $\le$ \textit{fu}$_{jX}$($X$, $T_j$) + $\sum_{x_i \in (T_j / X)}$\textit{fu}$_{ijl}$($x_i$, $T_j$) \\
		\> =  \textit{fu}$_{jX}$($X$, $T_j$) + \textit{rfu}($X$, $T_j$). \\
	\end{tabbing}
	Suppose $X^\prime.tids$ is the \textit{tid} set in list of $X^\prime$, and $X.tids$ that in $X$, then: \\
	$\because$ $X \subset X^\prime$ $\Rightarrow$ $X^\prime.tids \subseteq X.tids$
	\begin{tabbing}
		$\therefore$ \textit{fu}$_{X^\prime}$($X^\prime$) \=
		= $\sum_{T_j \in X^\prime.tids}$\textit{fu}$_{jX^\prime}$($X^\prime$, $T_j$) \\
		\> $\le$ $\sum_{T_j \in X^\prime.tids}$(\textit{fu}$_{jX}$($X$, $T_j$) + \textit{rfu}($X$, $T_j$)) \\
		\> $\le$ $\sum_{T_j \in X.tids}$(\textit{fu}$_{jX}$($X$, $T_j$) + \textit{rfu}($X$, $T_j$)) \\
		\> = \textit{sumIfu}($X$) + \textit{sumRfu}($X$).
	\end{tabbing}
	Therefore, if \textit{sumIfu}($X$) + \textit{sumRfu}($X$) $\le$ $\gamma$ that $X^\prime$ must be a low fuzzy utility itemset.
\end{proof}

For example, consider the 2-itemset \{\textit{B,Middle}, \textit{A.Low}\} is the extension of 1-itemset \{\textit{A.Low}\}. Since \textit{sumIfu}(\{\textit{A.Low}\}) + \textit{sumRfu}(\{\textit{A.Low}\}) = \$96.8 ($>$ $\gamma$), the extension \{\textit{B,Middle}, \textit{A.Low}\} in the enumeration tree is necessary to be generated. However, as for the 2-itemset \{\textit{B,Middle}, \textit{A.Low}\}, the summation of its \textit{ifu} and \textit{rfu} is lower than $\gamma$. Thus, the extension of 2-itemset is no need to search deeply. So far, we have introduced all key definitions and useful theorems, the detailed description of the FUIM algorithm proposed in this paper will be mentioned in the following.

Inspired by Ref. \cite{krishnamoorthy2015pruning}, if we study the remaining fuzzy utility constraint furthermore, there would be a lot of useless fuzzy utility values are computed. We continue consider the 2-itemset $\{$\textit{B,Middle}, \textit{A.Low}$\}$ as sample. In fuzzy-list of $\{$\textit{A.Low}$\}$, it contains five transactions ($T_2$, $T_4$, $T_5$, $T_6$ and $T_7$), and the extended fuzzy-list of $\{$\textit{B,Middle}$\}$ includes seven transactions ($T_1$, $T_2$, $T_3$, $T_4$, $T_6$, $T_7$ and $T_8$). Then the common transactions are $T_2$, $T_4$, $T_5$, $T_6$ and $T_7$, and we only need to take these transactions into account. Thus, $T_2$ and $T_5$ are two useless transactions when computing fuzzy values of $\{$\textit{B,Middle}, \textit{A.Low}$\}$. Herein, we propose a new constraint which is a tighter upper-bound than remaining fuzzy utility.

\begin{theorem}
	\textbf{(The expended fuzzy utility constrain)}
	\label{theo:sumExU}
	\rm Given two fuzzy itemsets $X$ and $Y$, if $\sum_{T_j \in ful(X)}$(\textit{ifu}($X$, $T_j$) + \textit{rfu}($X$, $T_j$)) - $\sum_{T_j \in ful(X) \land Y \not\subseteq T_j}$(\textit{ifu}($X$, $T_j$) + \textit{rfu}($X$, $T_j$)) $<$ $\gamma$, then any extension itemset based on $XY$ cannot be a HFUI.
\end{theorem}

\begin{proof}
	Assume the extension of fuzzy itemsets $X$ and $Y$ are $X^\prime$ and $Y^\prime$ respectively, and the extension of their super-itemset $XY$ is $X^\prime$$Y^\prime$.
	\begin{tabbing}
		$\because$ $\sum_{T_j \in ful(X)}$(\textit{ifu}($X$, $T_j$) + \textit{rfu}($X$, $T_j$)) \= \\
		= $\sum_{T_j \in ful(X) \land Y \subseteq T_j}$(\textit{ifu}($X$, $T_j$) + \textit{rfu}($X$, $T_j$)) \\
		+ $\sum_{T_j \in ful(X) \land Y \not\subseteq T_j}$(\textit{ifu}($X$, $T_j$) + \textit{rfu}($X$, $T_j$)) \\
		$\therefore$ $\sum_{T_j \in ful(X) \land Y \subseteq T_j}$(\textit{ifu}($X$, $T_j$) \= \\
		= $\sum_{T_j \in ful(X)}$(\textit{ifu}($X$, $T_j$) + \textit{rfu}($X$, $T_j$)) \\ - $\sum_{T_j \in ful(X) \land Y \not\subseteq T_j}$(\textit{ifu}($X$, $T_j$) + \textit{rfu}($X$, $T_j$)) \\
		= $\sum_{T_i \in ful(XY)}$(\textit{ifu}($XY$, $T_j$) + \textit{rfu}($XY$, $T_j$)).
	\end{tabbing}
	\begin{tabbing}
		$\because$ $T_k \subseteq T_i \subseteq T_j$ and $X^\prime$$Y^\prime$ is the extension of $XY$ \\
		$\therefore$ \textit{fu}($X^\prime$$Y^\prime$) \= =$\sum_{T_k \in ful(X^\prime Y^\prime)}$\textit{fu}$_{kX^\prime Y^\prime}$($X^\prime$$Y^\prime$, $T_k$) \\
		\> $\le$ $\sum_{T_i \in ful(XY)}$(\textit{ifu}($XY$, $T_i$) + \textit{rfu}($XY$, $T_i$)) \\
		\> = $\sum_{T_j \in ful(X)}$(\textit{ifu}($X$, $T_j$) + \textit{rfu}($X$, $T_j$)) \\
		\> - $\sum_{T_j \in ful(X) \land Y \not\subseteq T_j}$(\textit{ifu}($X$, $T_j$) + \textit{rfu}($X$, $T_j$)).
	\end{tabbing}
	Thus, Theorem \ref{theo:sumExU} is proofed.
\end{proof}

\subsection{Main FUIM Algorithm}

According to the aforementioned definitions, we had discussed the details of proposed algorithm. The whole framework of the FUIM algorithm is shown in Fig. \ref{fig:flowchart}. It takes quantitative transactions, external utility of items membership function and user-specified threshold as input, then traditional items are transformed into fuzzy items. Next, the key information of updated quantitative transactions will be extracted and forms efficient fuzzy-lists. After a series mining steps, a complete set of high fuzzy utility itemsets is discovered. We introduce the details of mining process as follows.

\begin{figure}[hbtp]
	\centering
	\includegraphics[width=0.49\textwidth,scale=0.6]{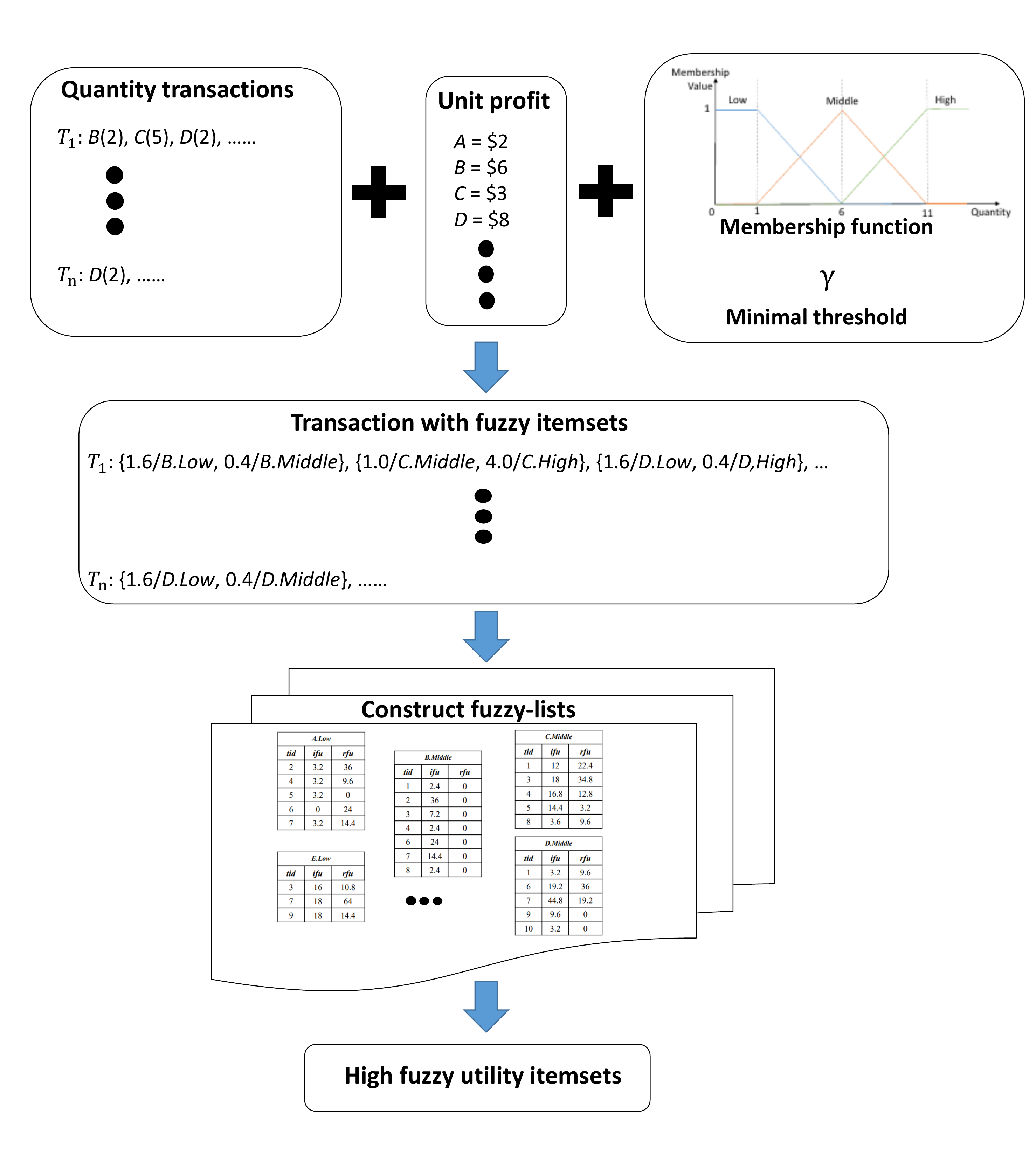}
	\caption{The framework of FUIM.}
	\label{fig:flowchart}
\end{figure}

We firstly introduce the main algorithm, in the pseudo-code of \textbf{Algorithm} \ref{algo:fuim_algorithm}, the input parameters include: 1) a quantitative transaction database $\mathcal{D}$, 2) a membership function $R$, and 3) a minimum fuzzy utility threshold $\gamma$. The output is a complete set of high fuzzy utility itemsets (HFUIs) in $\mathcal{D}$. In Lines 1-5, FUIM retrieves each transaction $T_j$ in $\mathcal{D}$ to collect key information about fuzzy utility upper-bound values of all items $x_i$ (\textit{fuub}($x_i$)). Then it obtains a global order of all $x_i \in I$ and revises items of each $T_j$ in $\mathcal{D}$ with the fuzzy utility upper-bound ascending order (Line 6). After that, common prefix itemset and its fuzzy-list are set as \textit{NULL}, which help to construct high level fuzzy itemsets (Lines 7 and 8). In Lines 9-12, it traverses all $x_i$. According Theorem \ref{theo:upperbound}, if the fuzzy utility upper-bound is no less than $\gamma$, that it is unnecessary to construct fuzzy-list; otherwise, it constructs fuzzy-list of HFUUBIs (\textit{fuList}$_1$). Then, in Line 14 this algorithm sets the selected \textit{fuList}$_1$ and \textit{prefix} as input parameters into \textbf{Algorithm} \ref{algo:miner_algorithm}. Finally, a complete set of HFUIs will be output (Line 15).

\begin{algorithm}[h]
	\caption{The FUIM algorithm}
	\label{algo:fuim_algorithm}
	\LinesNumbered
	\KwIn{$\mathcal{D}$: a quantitative transaction database; $R$: a membership function; $\gamma$: a user-specified minimum fuzzy utility threshold.}
	\KwOut{a complete set of fuzzy high utility itemsets (HFUIs).}
	
	\For{each transaction $T_j$ in $\mathcal{D}$} {
		convert the utility of all items $x_i \in I$ in $T_j$ to fuzzy utility value by $R$\;
		compute the maximum transaction fuzzy utility of $T_j$ (\textit{mtfu}$_j$)\;
		calculate the fuzzy utility upper-bound value of $x_i$ (\textit{fuub}($x_i$))\;
	}
	
	sort all $x_i \in I$ with the fuzzy utility upper-bound ascending order then get revised transactions\;
	initial common prefix itemset $P$ $\leftarrow$ \textit{NULL}\;
	initial fuzzy-list of $P$, \textit{fuList}$_{p}$ $\leftarrow$ \textit{NULL}\;
	
	\For{each item $x_i$ in $I$} {
		\If{\textit{fuub}($x_i$) $\ge$ $\gamma$} {
			get fuzzy-list of $x_i$ (\textit{fuList}$_1$)\;
		}
	}
	
	call \textbf{Miner}($P$, \textit{fuList}$_{p}$, \textit{fuList}$_1$, $\gamma$)\;
	
	\textbf{return} a complete set of HFUIs
\end{algorithm}

The details of the improved algorithm are shown in \textbf{Algorithm} \ref{algo:miner_algorithm}, and it is also an iteration method. It takes four parameters (a common prefix fuzzy itemset, the fuzzy-list of prefix, a set of fuzzy-lists, and a threshold) to discover all HFUIs constantly. For each \textit{fuList}$_X$ $\in$ \textit{fuLists}, it firstly checks whether $X$ is a HFUI or not (Lines 2-4). Then, according to Theorem \ref{theo:sumRfu}, if the summation of \textit{ifu} and \textit{rfu} of $X$ is no less than $\gamma$, it expends $X$ to construct more high level fuzzy itemsets (Lines 5-14). In Line 6, it uses \textit{exfuLists} to collect new extension fuzzy-lists, which is \textit{NULL} at the beginning. Since all 1-itemset are sorted with the ascending order of fuzzy utility upper-bounds, it just needs to consider these lists \textit{fuList}$_Y$ after \textit{fuList}$_X$ (Line 7). According to Theorem \ref{theo:sumExU}, if its inequality is false, and then it will call \textbf{Algorithm} \ref{algo:construct_lsit} to get new high level \textit{fuList} (Lines 8-13). At last, this procedure sets new parameters and iterates \textbf{Algorithm} \ref{algo:miner_algorithm} until there is no HFUIs are found (Lines 15 and 16).

\begin{algorithm}[h]
	\caption{The Miner function}
	\label{algo:miner_algorithm}
	\LinesNumbered
	\KwIn{$P$: a common prefix fuzzy itemset; \textit{fuList}$_P$: the fuzzy-list of $P$; \textit{fuList}s: a set of fuzzy-lists; $\gamma$: a user-specified minimum fuzzy utility threshold.}
	
	\For{each \textit{fuList}$_X$ in \textit{fuList}s} {
		\If{\textit{sumIfu}($X$) $\ge$ $\gamma$} {
			HFUIs $\leftarrow$ $X$\;
		}
		
		\If{\textit{sumIfu}($X$) + \textit{sumRfu}($X$) $\ge$ $\gamma$} {
			initial \textit{exfuList}s which is a new set of extended fuzzy-lists as \textit{NULL}\;
			
			\For{each \textit{fuList}$_Y$ after \textit{fuList}($X$) in \textit{fuList}s} {
				\If{\textit{sumIfu}($X$$Y$) + \textit{sumRfu}($X$$Y$) $\ge$ $\gamma$} {
					new \textit{fuList}$_{tmp}$ $\leftarrow$ call \textbf{Construct}(\textit{fuList}$_P$, \textit{fuList}$_X$, \textit{fuList}$_Y$)\;
					
					\If{\textit{sumIfu}(\textit{fuList}$_{tmp}$) $>$ 0} {
						add \textit{fuList}$_{tmp}$ into \textit{exfuList}s\;
					}
				}
			}
			
			$P$ $\leftarrow$ $P$ $\cup$ $X$\;
			call \textbf{Miner}($P$, \textit{fuList}$_X$, \textit{exfuList}s, $\gamma$)\;
		}
	}
\end{algorithm}

\textbf{Algorithm} \ref{algo:construct_lsit} takes three fuzzy-lists as input parameters. It combines two distinct lists to construct a high level list. In fact, as shown in Fig. \ref{fig:enumeratetree}, it joins two different lists if only if they have a common prefix. For the lists of 1-itemsets, their common prefix is assumed as \textit{NULL}. Line 1 initializes the fuzzy-list of new extension fuzzy itemset \textit{Pxy}. Then it traverses each element $Px_e$ in fuzzy-list of $Px$ (Line 2). If there exists an element $Py_e$ in \textit{fuList}$_{Py}$ and has the same \textit{tid} term of $Px_e$, two fuzzy itemsets $Px$ and $Py$ can form a new fuzzy itemset \textit{Pxy} (Line 5 and 6). Furthermore, if $Px$ and $Py$ is 1-itemset, it constructs a new element of \textit{Pxy} fuzzy-list, then calculates the fuzzy internal utility and remaining fuzzy utility of \textit{Pxy} (Line 8). On the contrary, the \textit{ifu} values should minus \textit{ifu}($P_e$) because of computing repeatedly (Line 6). Due to the quantitative transaction is already revised, the \textit{rfu} value of \textit{Pxy} should be \textit{rfu}($P_y$). In Line 10, a new element of fuzzy-list of \textit{Pxy} is constructed. In the end, \textbf{Algorithm} \ref{algo:construct_lsit} returns a new \textit{fuList}(\textit{Pxy}) (Line 13).

\begin{algorithm}
	\caption{The Construct function}
	\label{algo:construct_lsit}
	\LinesNumbered
	\KwIn{\textit{fuList}$_{P}$: the fuzzy-list of co-prefix fuzzy itemset; \textit{fuLsit}$_{Px}$: the fuzzy-list of fuzzy itemset $Px$; \textit{fuLsit}$_{Py}$: the fuzzy-list of fuzzy itemset $Py$.}
	\KwOut{a high level fuzzy-list of \textit{fuList}$_{Pxy}$.}
	
	initial \textit{fuList}$_{Pxy}$ $\leftarrow$ \textit{NULL}\;
	
	\For{each element $Px_e$ $\in$ \textit{fuList}$_{Px}$} {
		\If{$\exists Py_e$ $\in$ \textit{fuList}$_{Py}$ and $Px_e$\textit{.tid} == $Py_e$\textit{.tid}} {
			\eIf{\textit{fuList}$_{P}$ $\not=$ \textit{NULL}} {
				adopt binary search method find element $P_e$ $\in$ \textit{fuList}$_{P}$ which  $P_e$\textit{.tid} == $Px_e$\textit{.tid}\;
				$Pxy_e$ = ($Px_e$\textit{.tid}, \textit{ifu}($Pxy_e$)-\textit{ifu}($P_{e}$), \textit{rfu}($Py_e$))\;
			} {
				$Pxy_e$ = ($Px_e$\textit{.tid}, \textit{ifu}($Pxy_e$), \textit{rfu}($Py_e$))\;
			}
			add $Pxy_e$ into \textit{fuList}$Pxy$\;
		}
	}
	
	\textbf{return} a fuzzy-list of new fuzzy itemset $Pxy$
\end{algorithm}


\section{Experimental Results}
\label{sec:experimental}

In this section, we present extensive experiments to evaluate the effectiveness and efficiency of the proposed FUIM algorithm. To the best of our knowledge, the state-of-the-art study that most related to FUIM is TPFU \cite{lan2015fuzzy}. Since the fuzzy-based HUIM algorithm is much time costly, we select TPFU as a benchmark to evaluate the performance of FUIM, as well as the improved versions which adopted various theorems. In the following, FUIM$_1$, FUIM$_2$, FUIM$_3$, and FUIM respectively refer to different variant of the FUIM algorithm, with only Theorems \ref{theo:fuzzyDC} and \ref{theo:upperbound}, with Theorems \ref{theo:upperbound} and \ref{theo:sumRfu}, with Theorems \ref{theo:upperbound}, and \ref{theo:sumExU}, and with all theorems introduced. In addition, all algorithms have utilized Theorems \ref{theo:fuzzyDC} and \ref{theo:sumIfu}.

\subsection{Experimental Setup and Datasets}

We implemented all algorithms with Java language and conducted the experiments on a computer with an Intel Core 3.0 GHz processor with 16 GB of main memory running on Windows 10 Home Edition (64-bit operating system). To compare the performance of FUIM with TPFU, both real-life and synthetic datasets were used in experiments.

\textbf{Dataset description}. Six different datasets were used to test the efficiency of our novel algorithm, including four real datasets (foodmart, kosarak, mushroom, and retail) and two synthetic datasets (T10I4D100K and T40I10D100k). All the datasets have various characteristics and can be downloaded from the SPMF library\footnote{\url{http://www.philippe-fournier-viger.com/spmf/index.php}}. Then, we randomly generated unit utility of each item between 1 and 10,000 by a log-normal distribution method. At the same time, the quantity of each item was randomly generated from a pre-defined range (1-6) for each dataset (the last column in Table \ref{tab:experiments}). The statistical information of each dataset is given in Table \ref{tab:experiments}, including the number of transactions, the account of distinct items, the average and maximal length of a transaction, and the quantity range of item.

\begin{table}[!h]
	\begin{center}
		\caption{Statistical information about datasets}
		\label{tab:experiments}
		\begin{tabular}{lccccc}
			\hline
			\textbf{Dataset} & \#\textbf{Trans} & \#\textbf{Items} & \textbf{AvgLen} & \textbf{MaxLen} & \textbf{Range} \\ \hline 
			foodmart & 4141 & 1559 & 4.4 & 14 & 1-6 \\ 
			kosarak & 990002 & 41280 & 8.1 & 2498 & 1-6 \\ 
			mushroom & 8124 & 119 & 23 & 23 & 1-6 \\ 
			retail & 88162 & 16470 & 10.3 & 76 & 1-6 \\ 
			T10I4D100K & 10000 & 870 & 10.1 & 29 & 1-6 \\ 
			T40I10D100k & 10000 & 942 & 39.6 & 77 & 1-6 \\
			\hline
		\end{tabular}
	\end{center}
\end{table}

\textbf{Membership function}. As shown in Fig. \ref{fig:exMembership}, we supposed all items had the same membership function in the experiments. And there are three fuzzy regions (\textit{Low}, \textit{Middle}, and \textit{High}).

\begin{figure}[hbt]
	\centering
	\includegraphics[scale=0.5]{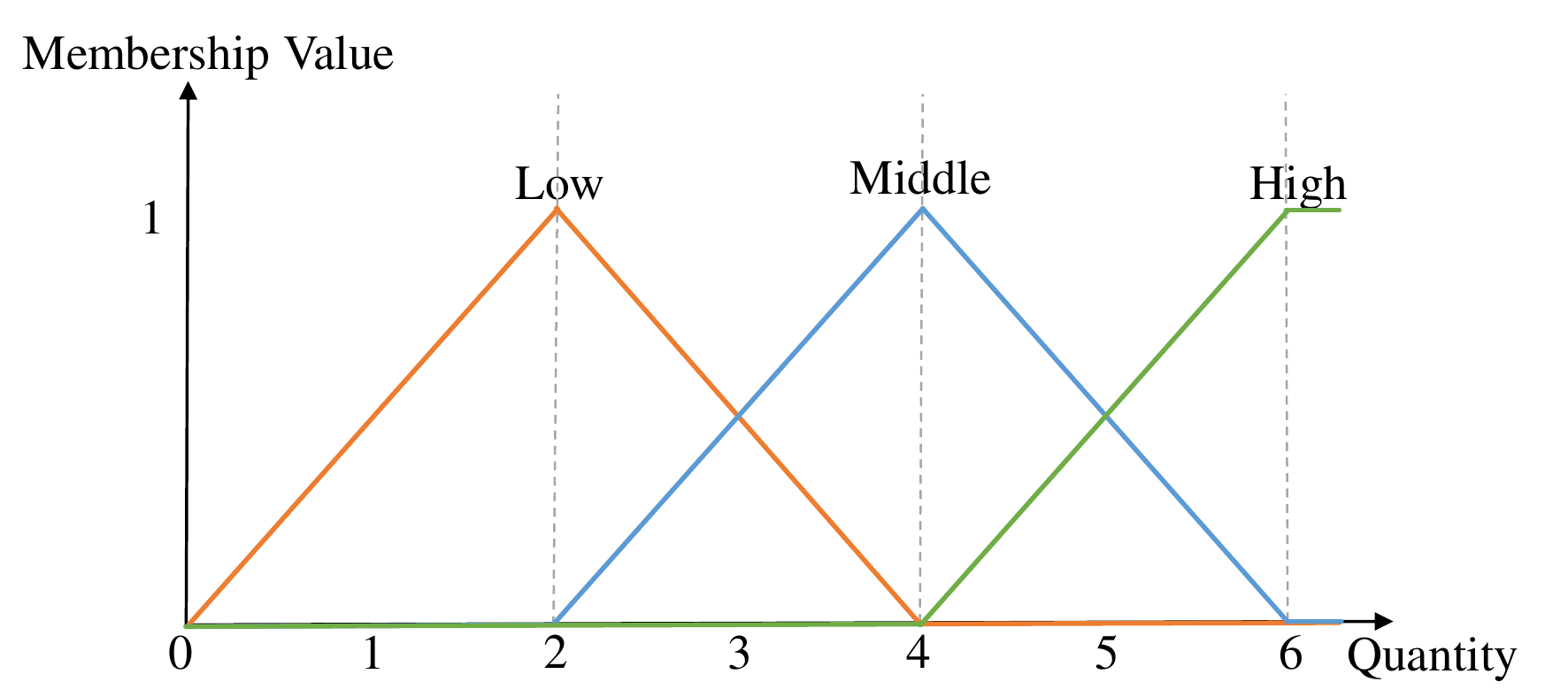}
	\caption{The membership function used in experiments.}
	\label{fig:exMembership}
\end{figure}

Note that three times of the experiments were executed for each test, and calculated their average runtime and memory consumption as final results. A wide range or the minimum fuzzy utility threshold ($\gamma$) was variously set for different datasets until a clear trend between the compared algorithms is revealed. In order to make results without loss of generality, for each dataset, we adopt threshold rate multiplies the total utility of whole dataset. The $\gamma$ is used to represent the threshold rate in the following content. In the designed process, we assume the algorithm is terminated if its runtime exceeds 10,000 seconds.

\subsection{Patterns and Candidates Analysis}

In order to analyze the relationship between fuzzy-based HUIM and traditional HUIM, the numbers of two types of discovered patterns (HFUIs and HUIs) are compared with the same $\gamma$. Notice that HFUIs are generated by TPFU (the state-of-the-art algorithm) and FUIM, and the HUIs are derived by the HUI-Miner algorithm. In Fig. \ref{fig:compare_HFUIs_with_HUIs}, we had tested foodmart dataset with various thresholds. Even though fuzzy-list may be produced much because of huge items, in fact, the membership function has limited the number of high level itemsets. Therefore, the amount of HUIs has a large majority over HFUIs. For example, the number of HFUIs is nearly half of HUIs when $\gamma$ = 0.001. Absolutely,  the smaller threshold we set, the greater gap will be. Then we can also safely infer the situation will be same appeared on other datasets.

\begin{figure}[hbt]
	\centering
	\includegraphics[scale=0.35, width=\linewidth]{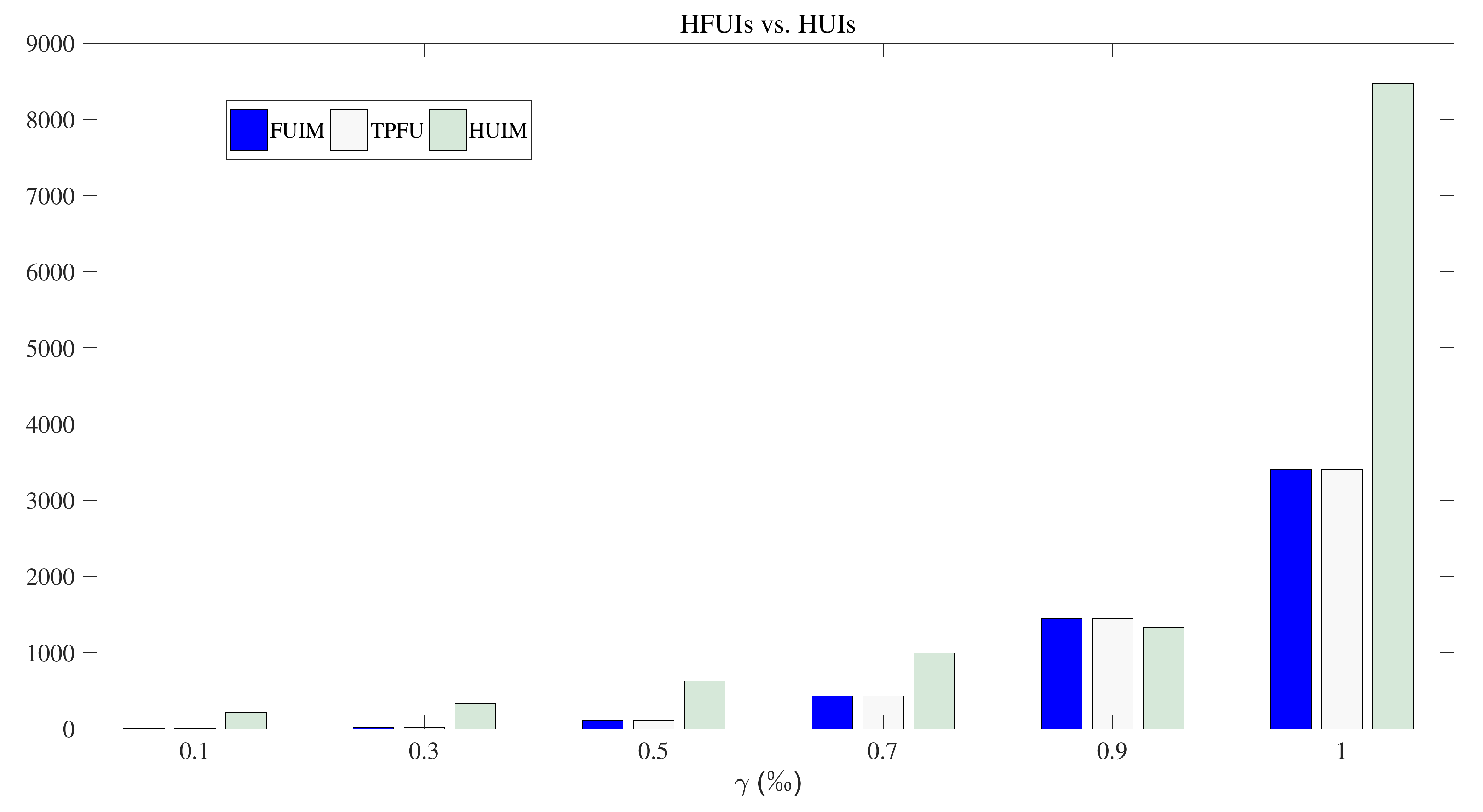}
	\caption{The comparison between HFUIs and HUIs.}
	\label{fig:compare_HFUIs_with_HUIs}
\end{figure}

\begin{table*}[htbp]
	\normalsize
	\centering
	\caption{Compare the number of candidates generating}
	\label{tab:compare_candidates}
	
	\begin{tabular}{|c|c|llllll|}
		\hline\hline
		\textbf{Dataset} & \textbf{Algorithm} & \multicolumn{6}{c|}{\textbf{\# Candidates when varying $\gamma$}} \\ \hline
		
		&$\gamma$	   & 1\textperthousand & 2\textperthousand & 3\textperthousand & 4\textperthousand & 5\textperthousand & - \\ \cline{2-8}
		&\textbf{TPFU} 		    & 207,042,480 & 98,615,784 & 1,796,415 & 173,612 & 173,612 & -\\
		\textbf{foodmart} &\textbf{FUIM}  		  & 5,725,017 & 1,218 & 119 & 0 & 0 & -\\
		&\textbf{FUIM$_1$}  & 31,672,791 & 18,387,668  & 83,480 & 0 & 0 & -\\
		&\textbf{FUIM$_2$}  & 6,858,907 & 5,725,017 & 1218 & 0 & 0 &- \\
		&\textbf{FUIM$_3$}  & 10,912,794 & 6,675,877 & 63,308 & 0 & 0 &- \\ \hline
		
		& $\gamma$	   & 1.2\% & 1.4\% & 1.6\% & 1.8\% & 2\% & 2.2\% \\ \cline{2-8}
		&\textbf{TPFU} 		    & - & - & - & - & - & - \\
		\textbf{mushroom} &\textbf{FUIM}  		   & 176,163 & 124,141 & 91,371 & 69,615 & 52,942 & 40,722 \\
		&\textbf{FUIM$_{1}$}  & - & - & - & - & - & -  \\
		&\textbf{FUIM$_{2}$}  & 268,150 & 181,638 & 132,503 & 101,968 & 77,833 & 59,858  \\
		&\textbf{FUIM$_{3}$}  & 330,511 & 232,521 & 169,156 & 130,923 & 105,155 & 87,068 \\ \hline
		
		& $\gamma$	    & 2\textperthousand & 2.5\textperthousand & 3\textperthousand & 3.5\textperthousand & 4\textperthousand & 4.5\textperthousand \\ \cline{2-8}
		&\textbf{TPFU} 		   & - & - & - & - & - & - \\
		\textbf{kosarak} &\textbf{FUIM}  		  & 43,288,292 & 32,926,784 & 22,022,584 & 10,830,174 & 4,531,538 & 2,192,118 \\
		&\textbf{FUIM$_{1}$}  & - & - & - & - & - & -  \\
		&\textbf{FUIM$_{2}$}  & 64,646,254 & 33,593,460 & 22,187,694 & 10,854,457 & 4,531,624 & 2,192, 187  \\
		&\textbf{FUIM$_{3}$}  & 114,461,240 & 91,518,312 & 79,905,639 & 71,772,428 & 64,727,735 & 58,696,947 \\ \hline
		
		& $\gamma$	   & 0.8\textperthousand & 1\textperthousand & 1.2\textperthousand & 1.4\textperthousand & 1.6\textperthousand & 1.8\textperthousand \\ \cline{2-8}
		&\textbf{TPFU} 		   & - & - & - & - & - & - \\
		\textbf{T10I4D100K} &\textbf{FUIM}  		  & 2,179,851 & 1,974,036 & 1,788,221 & 1,628,388 & 1,421,094 & 1,242,777 \\
		&\textbf{FUIM$_{1}$}  & - & - & - & - & - & -  \\
		&\textbf{FUIM$_{2}$}  & 2,722,923 & 2,205,591 & 1,883,173 & 1,664,747 & 1,432,227 & 1,247,622  \\
		&\textbf{FUIM$_{3}$}  & 2,868,034 & 2,712,449 & 2,653,794 & 2,570,223 & 2,488,714 & 2,408,797 \\ \hline
		
		\hline\hline
	\end{tabular}
\end{table*}

Furthermore, the candidates generated in TPFU and the visited nodes in FUIM are also compared to access the pruning effect. Thus, the detailed numbers of all tested algorithms are shown in Table \ref{tab:compare_candidates}. FUIM adopts the compact fuzzy-list data structure and effective pruning strategies that can reduce the searching space as much as possible, while TPFU is still a level-wise algorithm. Overall, it can declare FUIM is much faster than TPFU apparently, and TPFU often runs out of time in many experiments. More intuitively, when the $\gamma$ increases, the number of candidates is decreasing instead. The reason is that the pruning strategies may be easy to filter out more candidates with a larger threshold. For example, the visited nodes are ascending from 40,722 to 176,163 when $\gamma$ reduces from 2.2\% to 1.2\% on mushroom dataset. In addition, because FUUB is a loose upper bound, FUIM$_1$ is also often overtime in experiments. It should be figure out that, on foodmart dataset, the abnormal results are caused by overlarge threshold when $\gamma$ is 4\textperthousand $\,$  (= \$417,817) and 5\textperthousand $\,$  (= \$522,272). In a word, FUIM performs the best among the five tested algorithms, and it generates the least amount of candidates.

\subsection{Memory Usage Analysis}

Whether the fuzzy-list-based model is less memory cost than the two-phase Apriori-based model? To answer this question, we further carried experiments to assess the memory consumption of the compared algorithms using the same parameters. The details of the maximum memory usage for each algorithm are shown in Fig. \ref{fig:compare_memory}. On foodmart dataset, the average memory usage of TPFU is nearly eight times of that of FUIM, FUIM$_1$, FUIM$_2$, and FUIM$_3$. In particularly, due to the poor efficiency of TPFU, it often runs out of time that we assume it cannot get result on other datasets. Similarly, FUIM$_1$ faces the same troubles. In addition, the memory consumption of FUIM and FUIM$_{3}$ increase smoothly as $\gamma$ decrease. In fact, several effective pruning strategies are utilized in FUIM, and a compact data structure can store complete necessary message. This explains why FUIM always costs less memory then TPFU. In conclusion, FUIM performs very well as expected in experiments.

\begin{figure*}[!htbp]
	\centering
	\includegraphics[trim=10 0 10 0,clip,height=0.17\textheight,width=1\textwidth,scale=0.6]{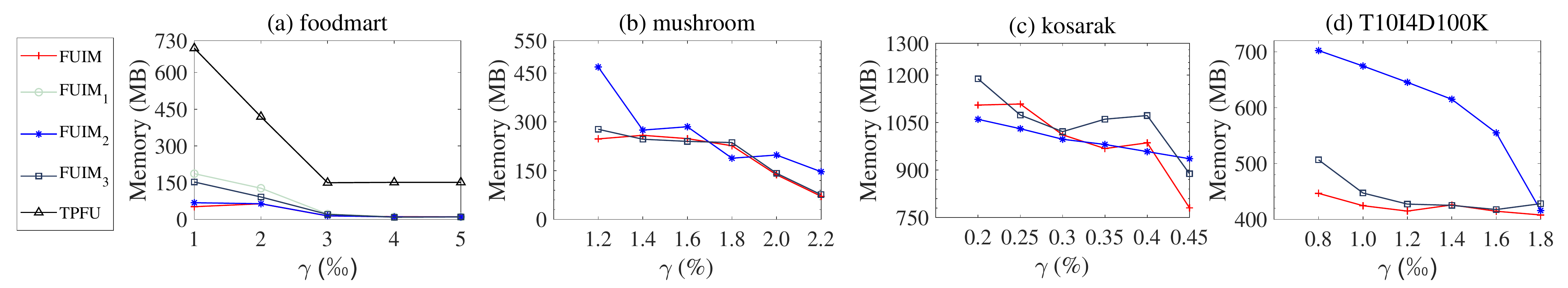}
	\caption{The memory consumption.}
	\label{fig:compare_memory}
\end{figure*}

\begin{figure*}[!htbp]
	\centering
	\includegraphics[trim=10 0 10 0,clip,height=0.17\textheight,width=1\textwidth,scale=0.6]{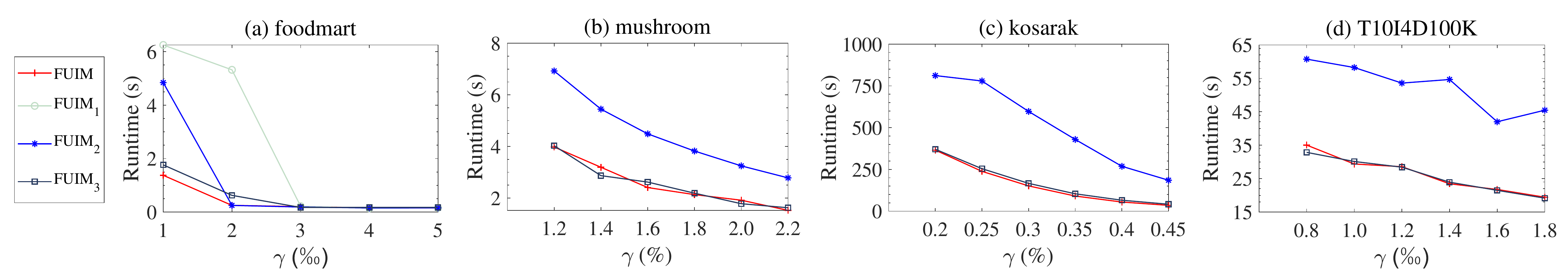}
	\caption{The runtime consumption.}
	\label{fig:compare_runtime}
\end{figure*}

\subsection{Runtime Analysis}

In this subsection, we compare the runtime of four versions of the proposed FUIM algorithm (FUIM$_1$, FUIM$_2$, UFIM$_3$, and FUIM) with TPFU. However, on the account of TPFU took so much time that makes an enormous gap with FUIM, we do not draw the line of TPFU. This is the common shortcoming of level-wise algorithms, whereby TPFU generated a large amount of useless candidates. Fig. \ref{fig:compare_runtime} shows the runtime consumption of four compared algorithms under various $\gamma$ on different datasets. Absolutely, the proposed FUIM algorithm significantly outperformed others in all cases. Because of the same reason, FUIM$_2$ also only gets result on foodmart dataset. In addition, FUIM$_{3}$ not only has advantages over FUIM$_{2}$ in runtime usage, but also performs well in memory consumption. This indicates that Theorem \ref{theo:fuzzyDC} with the tighter upper bound is more suitable than Theorem \ref{theo:sumRfu} and its upper bound. For example, on kosarak dataset, the average runtime consumption of FUIM$_3$ is nearly one-in-three of  FUIM$_{2}$ spends. In conclusion, FUIM can effectively mine correct HFUIs within reasonable time according to outstanding theorems.

\subsection{Processing Order of Items Analysis}

In general, the chosen processing order of items usually may influence the performance of a data mining algorithm. As mentioned previously, the final results of FUIM are complete and correct, while the mining performance maybe influenced. Therefore, the effect of different processing order of items is evaluated in this subsection. The results of two orders (including the FUUB-descending order and the FUUB-ascending order) on retail dataset are shown in Fig. \ref{fig:compare_order}. The symbol $N_1$ represents the FUUB-ascending order, and FUUB-decreasing order is denoted as $N_2$. Clearly, $N_1$ leads to the better performance in terms of runtime and the account of visited nodes. For example, when $\gamma$ is 0.1\textperthousand, $N_2$ takes more 300 seconds than $N_1$. Although the memory consumption of $N_1$ is often higher than $N_2$, the gap between $N_1$ and $N_2$ is narrow when threshold decreasing. Thus, it can be concluded that the choice of distinct processing order of items influences the number of fuzzy-lists generating for the mining task. Furthermore, we can determine that the FUUB-ascending order is acceptable in our novel algorithm for mining HFUIs.

\begin{figure}[!hbt]
	\centering
	\includegraphics[trim=10 0 10 0,clip,width=0.5\textwidth,scale=0.6]{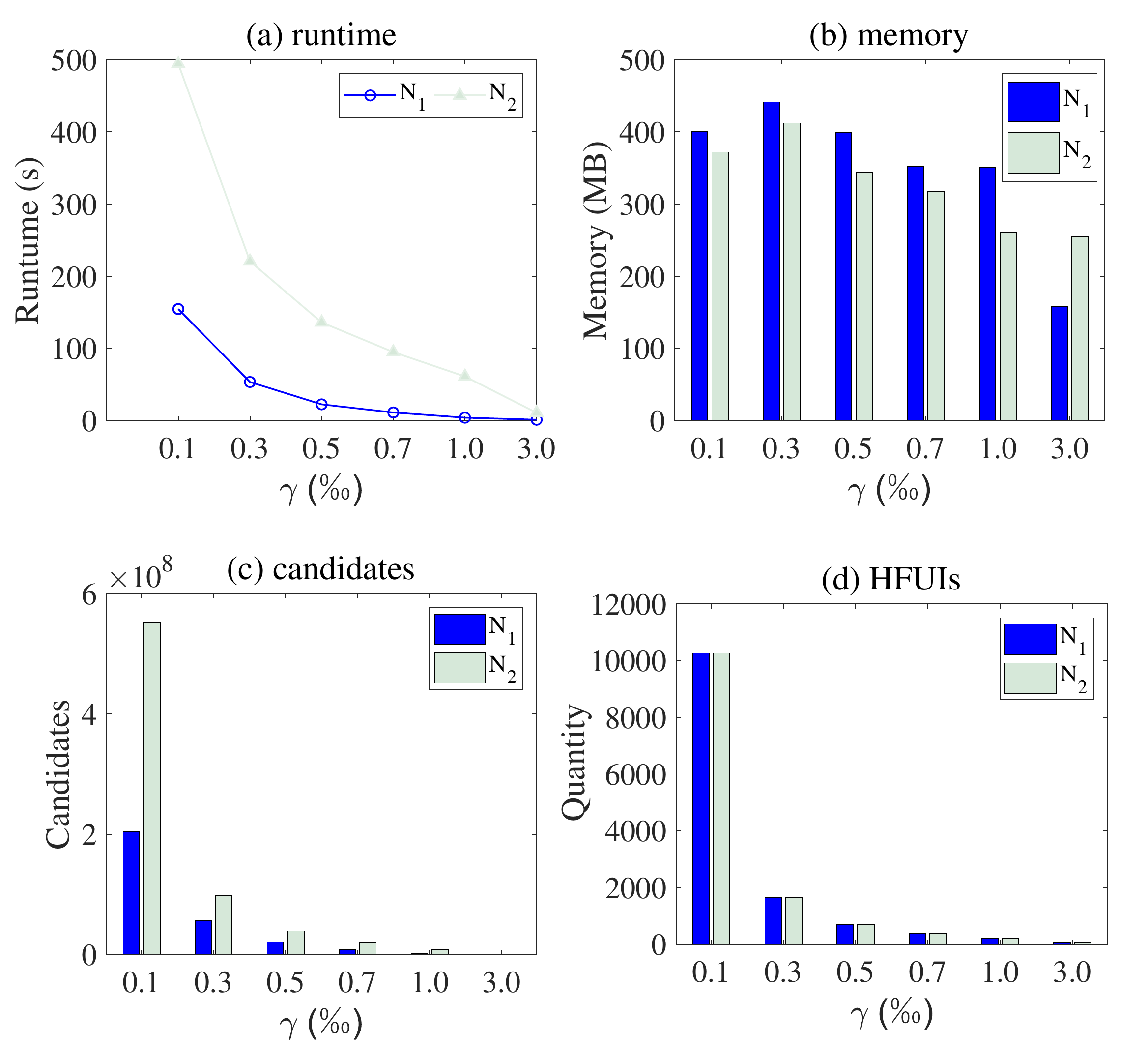}
	\caption{The comparison of two different orders about fuzzy items.}
	\label{fig:compare_order}
\end{figure}

\begin{figure}[!htbp]
	\centering
	\includegraphics[trim=10 10 10 0,clip,height=0.18\textheight,width=0.4\textwidth,scale=0.46]{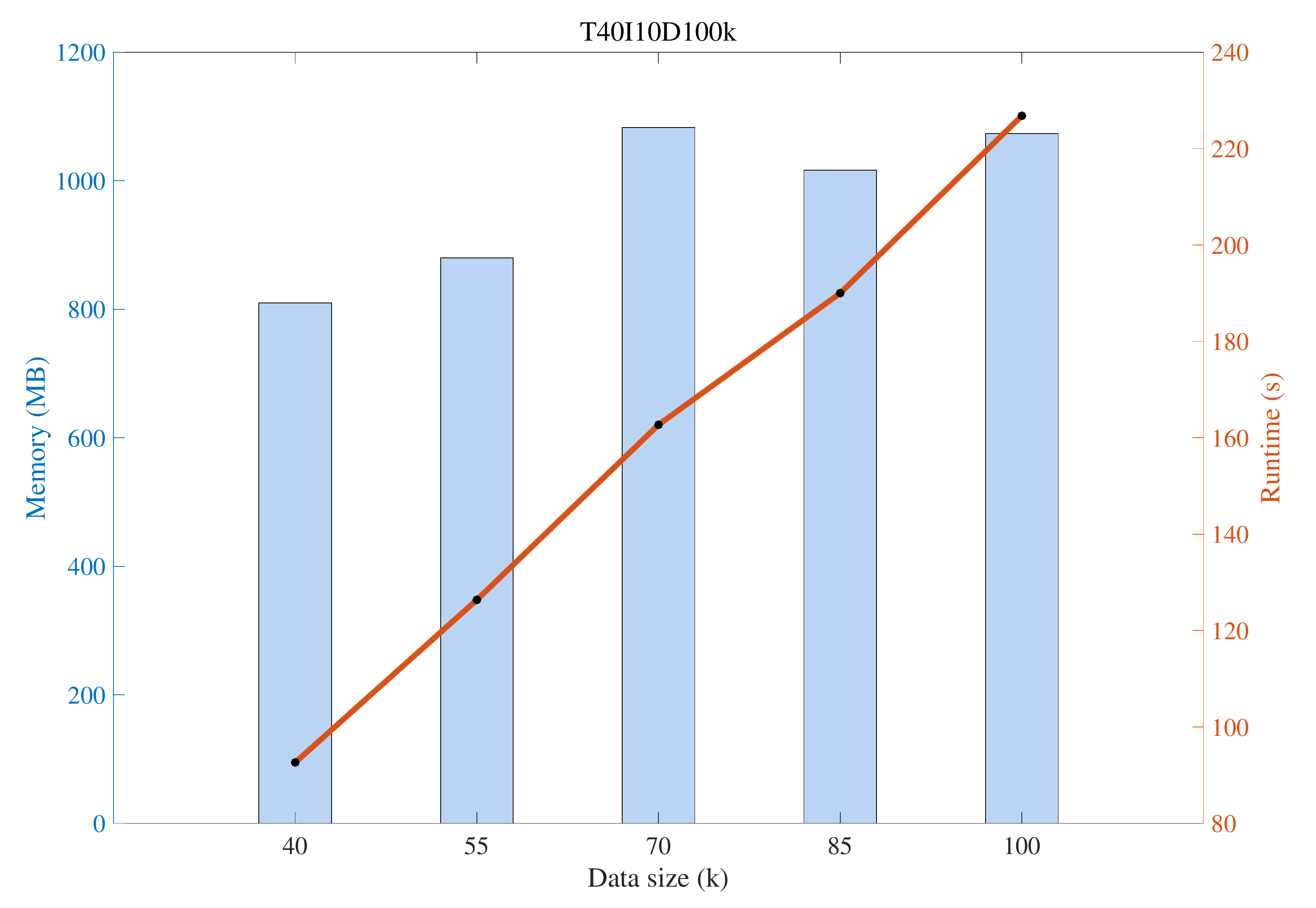}
	\caption{The scalability test.}
	\label{fig:compare_scalability}
\end{figure}

\subsection{Scalability Analysis}

The last experiment is the scalability test on a synthetic dataset T40I10D100k. In each test, the dataset size ($|D|$) is respectively 40k, 55k, 70k, 85k, and 100k transactions when $\gamma$ is fixed to 0.1\%. Fig. \ref{fig:compare_scalability} shows the results in terms of runtime usage and memory consumption. It can be seen that both the runtime and memory usage increase smoothly along with the dataset size increases. In conclusion, the FUIM algorithm has a well scalability.

In a word, FUIM significant outperforms the TPFU algorithm. For example, consider the execution time, FUIM can be up to least three orders of magnitude faster than TPFU. Aa a two-phase model, TPFU firstly generates a huge number of candidates using level-wise mechanism, then scans the dataset again to calculate the final fuzzy utility value of each candidate, and finally discovers final HFUIs. However, in our proposed FUIM, the remaining fuzzy utility is utilized to obtain a powerful upper-bound, which reduces many unpromising fuzzy itemsets and avoids to perform the depth-first searching deeply. Whatever on sparse or dense datasets (e.g., kosarak and mushroom), without remaining fuzzy utility pruning strategies, FUIM cannot get correct result within 10,000 seconds. Furthermore, the ``mining during constructing" property in FUIM can avoid consuming too much memory. For example, when $\gamma$ is set to 0.001, FUIM only consumes 51.84 MB memory, while the level-wise model TPFU requires 698.96 MB memory on foodmart dataset.

\section{Conclusion and Future Works}
\label{sec:conclusion}

This work presents a new fuzzy utility mining algorithm named FUIM to find complete and correct HFUIs in quantitative transaction datasets. FUIM utilizes the minimal operation mechanism to construct itemset and maximum operation mechanism to calculate upper-bounds on fuzzy utility. Different from the existing algorithms, evidently, an effective remaining upper-bound model is developed to reduce the searching space. The sufficient experimental results reveal the proposed fuzzy-list-based FUIM algorithm can perform well when working on both synthetic and real datasets under different thresholds. Fuzzy-driven utility mining is an interesting topic. In the future, we will further improve the efficiency of the FUIM algorithm, or propose more effective data structure than fuzzy-list. How to apply FUIM to address stream data mining  and on-shelf availability is still challenging.

\ifCLASSOPTIONcaptionsoff
  \newpage
\fi

\bibliographystyle{IEEEtran}
\bibliography{FUIM}

\end{document}